\newif\ifOneColumn 
\OneColumnfalse

\newif\ifLinksBlack 
\LinksBlackfalse

\ifOneColumn
	\documentclass[journal,12pt,draftclsnofoot,onecolumn]{IEEEtran}
\else
	\documentclass[journal]{IEEEtran}
\fi

\IEEEoverridecommandlockouts
\usepackage{cite}
\usepackage{amsmath,amsthm,amssymb,amsfonts}
\usepackage{graphicx}
\usepackage{bmpsize}
\usepackage{textcomp}
\usepackage[linesnumbered,lined,boxed,commentsnumbered, ruled]{algorithm2e} 

\usepackage{psfrag}
\newtheorem{theorem}{Proposition}

\DeclareMathOperator*{\argmin}{arg\,min}
\newcommand{\Tr}{Tr}

\newcommand{\I}{\mathbf{I}}

\newcommand{\A}{\mathbf{A}}

\newcommand{\Eu}{\mathbf{E}_{\mathrm{u}}}
\newcommand{\Ed}{\mathbf{E}_{\mathrm{d}}}
\newcommand{\E}{\mathbb{E}}
\newcommand{\e}{\mathbf{e}}
\newcommand{\D}{\mathbf{D}}
\newcommand{\Q}{\mathbf{Q}}

\newcommand{\Wbf}{\mathbf{W}}

\newcommand{\Pbf}{\mathbf{P}}
\newcommand{\Hbf}{\mathbf{H}}

\newcommand{\h}{\mathbf{h}}
\newcommand{\w}{\mathbf{w}}
\newcommand{\p}{\mathbf{p}}

\newcommand{\normht}{\|\mathbf{h}\|^{2}}
\newcommand{\normhf}{\|\mathbf{h}\|^{4}}

\newcommand{\SINRap}{\widetilde{\text{SINR}}}
\newcommand{\hhovernormh}{\frac{\h\h^{H}}{\normht}}

\newcommand{\hatx}{\hat{\mathbf{x}}}
\newcommand{\hatxu}{\hat{\mathbf{x}}^{\mathrm{u}}}

\newcommand{\yu}{\mathbf{y}^{\mathrm{u}}}
\newcommand{\xu}{\mathbf{x}^{\mathrm{u}}}
\newcommand{\yd}{\mathbf{y}^{\mathrm{d}}}
\newcommand{\xd}{\mathbf{x}^{\mathrm{d}}}
\newcommand{\xdd}{\tilde{\mathbf{x}}^{\mathrm{d}}}
\newcommand{\nup}{\mathbf{n}^{\mathrm{u}}}
\newcommand{\zu}{\mathbf{z}^{\mathrm{u}}}
\newcommand{\nd}{\mathbf{n}^{\mathrm{d}}}
\newcommand{\ymd}{y_{\textit{m}}^{\mathrm{d}}}
\newcommand{\ymu}{y_{\textit{m}}^{\mathrm{u}}}

\def\BibTeX{{\rm B\kern-.05em{\sc i\kern-.025em b}\kern-.08em
    T\kern-.1667em\lower.7ex\hbox{E}\kern-.125emX}}
\usepackage[caption=false,font=footnotesize]{subfig}

\ifLinksBlack
\RequirePackage[colorlinks,allcolors=black]{hyperref}
\else
\RequirePackage[colorlinks,allcolors=blue]{hyperref}
\fi

\graphicspath {{figures/}}

\begin{document}

\onecolumn 
\textcircled{c} 2020 IEEE.  Personal use of this material is permitted.  Permission from IEEE must be obtained for all other uses, in any current or future media, including reprinting/republishing this material for advertising or promotional purposes, creating new collective works, for resale or redistribution to servers or lists, or reuse of any copyrighted component of this work in other works.\\\\
\href{https://doi.org/10.1109/TSP.2020.2964496}{10.1109/TSP.2020.2964496}
\thispagestyle{empty} 
\twocolumn 
\setcounter{page}{1} 

\title{Decentralized Massive MIMO Processing Exploring Daisy-chain Architecture and Recursive Algorithms}

\author{
	Jes\'{u}s~Rodr\'{i}guez~S\'{a}nchez~\href{https://orcid.org/0000-0002-5531-1071}{\includegraphics[scale=0.04]{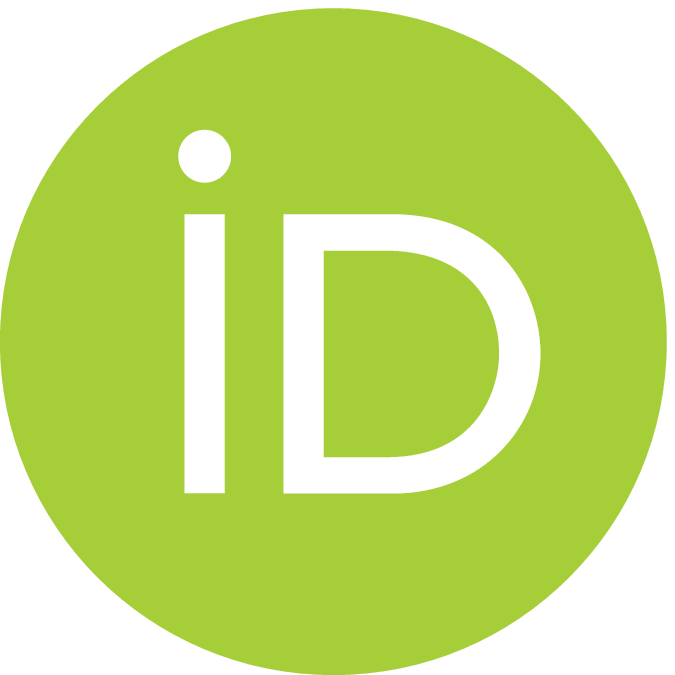}},~\IEEEmembership{Student~Member,~IEEE,}
	Fredrik~Rusek~\href{https://orcid.org/0000-0002-2077-3858}{\includegraphics[scale=0.04]{orcid.eps}},~\IEEEmembership{Member,~IEEE,}
	Ove~Edfors~\href{https://orcid.org/0000-0001-5966-8468}{\includegraphics[scale=0.04]{orcid.eps}},~\IEEEmembership{Senior~Member,~IEEE,}
	Muris~Sarajli\'{c}~\href{https://orcid.org/0000-0001-5107-8164}{\includegraphics[scale=0.04]{orcid.eps}},~\IEEEmembership{Student~Member,~IEEE,}
	and~Liang~Liu~\href{https://orcid.org/0000-0001-9491-8821}{\includegraphics[scale=0.04]{orcid.eps}},~\IEEEmembership{Member,~IEEE}
\thanks{This paper is build upon previous results presented at the 2018 SiPS Conference \cite{jesus} and in IEEE Wireless Communication Letters \cite{muris}.
The authors are with the Electrical and Information Technology, Lund University, 22363 Lund, Sweden (e-mail: \{jesus.rodriguez, fredrik.rusek, ove.edfors, muris.sarajlic, liang.liu\}@eit.lth.se).}
}

\maketitle
\begin{abstract}
	Algorithms for Massive MIMO uplink detection and downlink precoding typically rely on a centralized approach, by which baseband data from all antenna modules are routed to a central node in order to be processed. In the case of Massive MIMO, where hundreds or thousands of antennas are expected in the base-station, said routing becomes a bottleneck since interconnection throughput is limited. This paper presents a fully decentralized architecture and an algorithm for Massive MIMO uplink detection and downlink precoding based on the Coordinate Descent (CD) method, which does not require a central node for these tasks. Through a recursive approach and very low complexity operations, the proposed algorithm provides a good trade-off between performance, interconnection throughput and latency. Further, our proposed solution achieves significantly lower interconnection data-rate than other architectures, enabling future scalability.
\end{abstract}

\begin{IEEEkeywords}
	Massive MIMO, Decentralized Architecture, Precoding, Detection, zero-forcing, Coordinate Descent, Kaczmarz, inter-connection data-rate, SINR.
\end{IEEEkeywords}

\section{Introduction}
\label{section:intro}

Massive MIMO is one of the most relevant technologies in wireless communications \cite{marzetta,rusek}. Among the key features of this technology are high spectral efficiency and improved link reliability, making it a key enabler for 5G. Massive MIMO exploits spatial diversity far beyond traditional MIMO systems by employing a large scale antenna array in the base-station (BS) with hundreds or possibly even thousands of elements. This large number of elements allows for unprecedented spatial resolution and high spectral efficiency, while providing simultaneous service to several users within the same time-frequency resource.

Despite all the advantages of Massive MIMO, there are still challenges from an implementation point of view. One of the most critical ones is sending data from the BS antennas to the central processing unit (CPU) and vice-versa, and the high interconnection throughput it requires. In current set-ups, uplink detection algorithms based on zero-forcing (ZF) equalizer typically rely on a centralized architecture, shown in Fig. \ref{fig:BS_centralized}, where baseband samples are collected in the CPU for obtaining channel state information (CSI) and further matrix inversion, which allows data estimation and further detection. The same argument is valid for downlink precoding. In order to avoid dedicated links between antenna modules and CPU, a shared bus is typically used to exchange this data. In case of LuMaMi testbed \cite{lumami,lumami2}, the shared bus was reported to support an aggregated data-rate of 384Gps, which exceed base-station internal interface standards such as eCPRI \cite{ecpri}. Additionally, the pin-count of integrated circuits (IC) limits the number of links the IC can handle simultaneously and thus the throughput. Due to this high data-rate, the power appears as another potential limitation. This combination of factors are considered as the main bottleneck in the system and a clear limitation for array scalability. In this paper we will address the inter-connection throughput limitation by decreasing its value per link and consequently reducing the impact of the other two (pin-count and power).

The inter-connection bottleneck has been noted in several previous studies on different architectures for Massive MIMO BSs \cite{argos,Bertilsson,puglielli,lumami,cavallaro,li_jeon,jeon_li}. As a solution, most of these studies recommend moving to a decentralized approach where uplink estimation and downlink precoding can be performed locally in processing nodes close to the antennas (final detection can still be done in a CPU). However, to achieve that, CSI still needs to be collected in the CPU, where matrix inversion is performed \cite{argos,Bertilsson,lumami}, imposing an overhead in data shuffling.

The CSI problem is addressed in \cite{cavallaro}, where CSI is obtained and used only locally (not shared) for precoding and estimation, with performance close to MMSE. However, this architecture relies on the CPU for exchanging a certain amount of consensus information between the nodes, and this exchange negatively impacts the processing latency and throughput \cite{li_jeon}, and therefore limits the scalability of this solution. In order to solve these problems, feedforward architectures for detection \cite{jeon_li} and precoding \cite{li_jeon} have been proposed recently, where the authors present a partially decentralized (PD) architecture for detection and precoding, which achieves the same results as linear methods (MRC, ZF, L-MMSE), and therefore becomes optimal when $M$ is large enough. Partial Gramian matrices from antennas are added up before arriving to a processing unit where the Gramian is inverted.

In \cite{argos}, a flat-tree structure with daisy-chained nodes was presented. The authors propose conjugate beamforming as a fully decentralized method with the corresponding penalty in system capacity. In the same work it is also pointed out that by following this topology the latency was being severely compromised. The more detailed analysis on latency is thus needed to evaluate the algorithm.

In this article we propose a fully decentralized architecture and a recursive algorithm for Massive MIMO detection and precoding, which is able to achieve very low inter-connection data-rate without compromising latency.
The proposed algorithm is pipelined so that it runs in a distributed way at the antenna processing units, providing local vectors for estimation/detection that approximate to the zero-forcing solution.
We make use of the Coordinate Descent (CD) algorithm, which is detailed in Section \ref{section:CD}, to compute these vectors.

There is previous work based on CD, such as \cite{li_CD}. The main difference is that the coordinate update in \cite{li_CD} is done per user basis, i.e., a different user index is updated every iteration, while in our proposed method the coordinate update is done per antenna basis, updating all users at once.

We extend the work presented in \cite{jesus} and \cite{muris}, which are also based on decentralized daisy-chain architecture. The novelties of the present work compared to these two is as follows:
\begin{itemize}
\item A common strategy for downlink precoding and uplink equalization is presented, in contrast to \cite{jesus} and \cite{muris}, which only covers uplink and downlink separately.
\item The algorithm has been modified that serial processing is only needed when new CSIs are estimated. The corresponding filtering phase can be conducted in parallel to reduce latency, in contrast to \cite{jesus}, where serial processing is always needed, which increases the latency.
\item A recommended step-size is provided, in contrast to \cite{jesus}.
\item An analytical expression for resulting SINR and a complete performance analysis is presented in this paper.
\item Complexity analysis from a general point of view (not attached to any specific implementation) is provided, which includes: inter-connection data-rate, memory size and latency. In \cite{jesus}, only inter-connection data-rates are analyzed.
\end{itemize}

Decentralized architectures, as shown in Fig. \ref{fig:BS_decentralized}, have several advantages compared to the centralized counterpart, as shown in Fig. \ref{fig:BS_centralized}. For example, they overcome bottlenecks by finding a more equal distribution of the system requirements among the processing nodes of the system. Apart from this, data localization is a key characteristic of decentralized architectures. In uplink, the architecture allows data to be consumed as close as possible to where it is generated, minimizing the amount to transfer, and therefore saving throughput and energy. To achieve data localization, processing nodes need to be located near the antenna, where they perform processing tasks locally such as channel and data estimation. Local CSI is estimated and stored locally in each, without any need to share it with any other nodes in the system. This approach has been suggested previously in \cite{argos,Bertilsson,jeon_li,cavallaro,li_jeon,puglielli}, and we take advantage of it in the proposed solution.

The remainder of the paper is organized as follows. In Section \ref{section:background} the preliminaries are presented, comprising the system model for uplink and downlink, together with an introduction to linear processing and the ZF method. Section \ref{section:central_vs_decentral} is dedicated to a comparison between the centralized and decentralized architectures and reasoning why the latter one is needed, together with an overview of the daisy-chain topology. The proposed algorithm, based on CD, is presented in Section \ref{section:CD}. In \ref{section:analysis} closed-form expressions of the SIR and SINR are provided for this algorithm, together with  interconnection data-rates, latency and memory requirements of the proposed solution. Finally, Section \ref{section:conclusions} summarizes the conclusions of this publication. 

Notation: In this paper, lowercase, bold lowercase and upper bold face
letters stand for scalar, column vector and matrix, respectively. The
operations $(.)^T$, $(.)^*$ and $(.)^H$ denote transpose, conjugate and conjugate transpose respectively.
The $i$-th element of vector $\h$ is denoted as $h_{i}$. A vector $\w$ and a matrix $\A$ related to the $m$th antenna is denoted by $\w_m$ and $\A_{m}$, respectively. $A_{i,j}$ denotes element $(i,j)$ of $\A$. $\mathbf{A}_{m}(i,j)$ denotes element $(i,j)$ of the $m$-th matrix in the sequence $\{\A_{m}\}$. The $k$th coordinate vector in $\mathbb{R}^{K}$ is defined as $\e_{k}$. Kronecker delta is represented as $\delta_{ij}$. Probability density function and cumulative density function are denoted respectively as $f_{\mathbf{X}}(x)$ and $F_{\mathbf{X}}(x)$. Computational complexity is measured in terms of the number of complex-valued multiplications.\\

\section{Background}
\label{section:background}
\subsection{System model}
For uplink, we consider a scenario with $K$ single-antenna users transmitting to a BS with an antenna array with $M$ elements. Assuming time-frequency-based channel access, a Resource
Element (RE) represents a unit in the time-frequency grid (also
named subcarrier in OFDM) where the channel is expected to be approximately flat. Under this scenario, the input-output relation is
\begin{equation}
\yu = \Hbf\xu + \nup,
\label{eq:ul_model}
\end{equation}
where $\yu$ is the $M \times 1$ received vector, $\xu$ is the transmitted user data vector ($K \times 1$), $\Hbf=[\h_1 \; \h_2 \, \cdots \, \h_M]^{{T}}$ is the channel matrix ($M \times K$) and $\nup$ an $M \times 1$ vector of white, zero-mean complex Gaussian noise. The entries of $\Hbf$ are i.i.d. zero-mean circularly-symmetric complex-gaussian entries, with rows $\h_{i} \sim \mathcal{CN}(0, \I)$ for all $i$. The noise covariance at the receiver is  $N_{0}\I$. The average transmitted power is assumed to be equal across all users and we assume, without any loss of generality, a unit transmit power. SNR is defined as $\frac{1}{N_{0}}$ and represents the average "transmit" signal-to-noise ratio.

For downlink, if Time Division Duplex (TDD) is assumed, then according to channel reciprocity principle and by employing reciprocity calibration techniques \cite{joao}, it is assumed that within the same coherence time, the channel matrix is the same as in the uplink case, and the system model follows
\begin{equation}
\xdd = \Hbf^{T}\yd + \nd,
\label{eq:dl_model}
\end{equation}
for a RE, where $\yd$ is the $M \times 1$ transmitted vector, $\xdd$ is the received data vector by users ($K \times 1$), and $\nd$ samples of noise ($K \times 1$).

Once the system model is established, we introduce the linear processing fundamentals used for downlink precoding and uplink estimation.

\begin{figure*}\centering
	\footnotesize
	\subfloat[Centralized architecture]{
		\psfrag{1}{$1$}
		\psfrag{M}{$M$}
		\psfrag{RPU}[][][0.7]{$\mathrm{RPU}$}
		\psfrag{RF}[][][0.7]{$\mathrm{RF}$}
		\psfrag{OFDM}[][][0.55]{$\mathrm{OFDM}$}
		\psfrag{CPU}{$\mathrm{CPU}$}
		\psfrag{CHEST}[][][0.6]{$\mathrm{CHEST}$}
		\psfrag{EST}[][][0.6]{$\mathrm{EST}$}
		\psfrag{DET}[][][0.6]{$\mathrm{DET}$}
		\psfrag{DEC}[][][0.6]{$\mathrm{DEC}$}
		\psfrag{Bs}[][][1.0]{$\text{Base Station}$}
		\psfrag{Rc}{$R_\mathrm{c}$}	
		\includegraphics[width=0.35\textwidth]{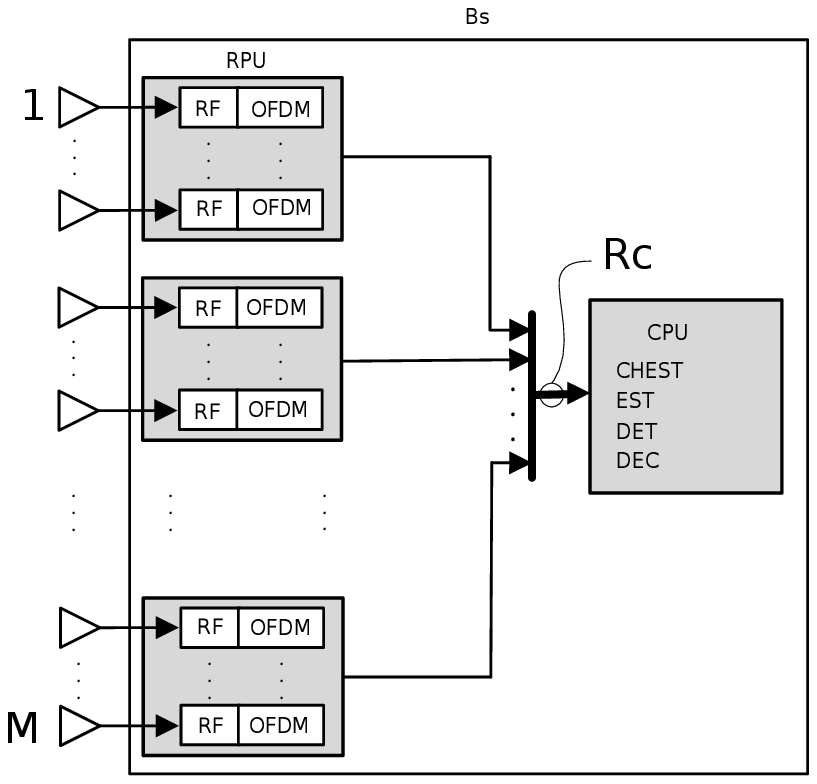}
		\label{fig:BS_centralized}
	}
	\subfloat[Decentralized architecture]{
		\psfrag{1}{$1$}
		\psfrag{M}{$M$}
		\psfrag{RPU}[][][0.7]{$\mathrm{RPU}$}
		\psfrag{RF}[][][0.7]{$\mathrm{RF}$}
		\psfrag{OFDM}[][][0.55]{$\mathrm{OFDM}$}
		\psfrag{CPU}{$\mathrm{CPU}$}
		\psfrag{CHEST}[][][0.5]{$\mathrm{CHEST}$}
		\psfrag{EST}[][][0.6]{$\mathrm{EST}$}
		\psfrag{DET}[][][0.6]{$\mathrm{DET}$}
		\psfrag{DEC}[][][0.6]{$\mathrm{DEC}$}
		\psfrag{Bs}{$\text{Base Station}$}
		\psfrag{Rd}{$R_\mathrm{d}$}
		\includegraphics[width=0.35\textwidth]{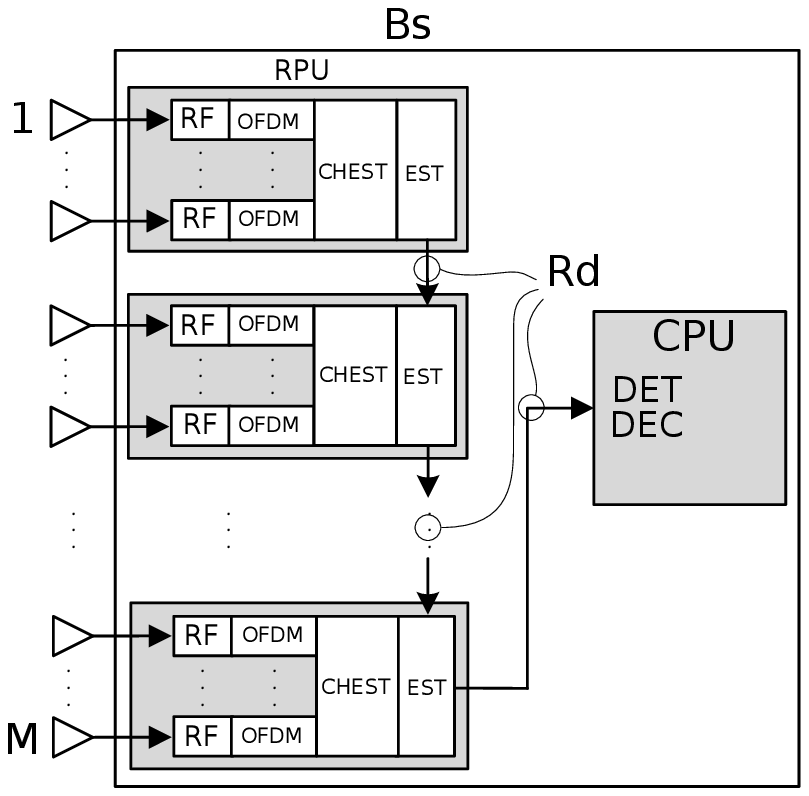}
		\label{fig:BS_decentralized}
	}
	
	\caption{Comparison between base station receiver chain in centralized and fully decentralized architectures for Massive MIMO uplink. Antenna array with $M$ elements is divided into RPUs, each containing a set of antennas. (a): Centralized architecture. Each RPU has one link to transfer baseband samples to the CPU, where the rest of processing tasks are done. (b): Fully decentralized architecture for detection. Each RPU performs RF, ADC, OFDM, channel estimation (CHEST) and data estimation (EST) locally. Detection (DET) and decoding (DEC) is centralized. RPUs are connected to each other by uni-directional links. Only one RPU has a direct connection with the CPU. Proposed algorithms are executed in EST blocks in parallel mode. The points where the interconnection data-rate is estimated are marked by circles and the value is denoted by $\mathrm{R}_{c}$ and $\mathrm{R}_{d}$ for centralized and decentralized respectively. The goal is to have  $\mathrm{R}_{d} \ll \mathrm{R}_{c}$ without compromising performance and latency.}
	\label{fig:comparison}
\end{figure*}

\subsection{Linear Processing}
In this article we focus on linear estimators and precoders, because they show close to optimal performance in Massive MIMO regime while requiring low complexity.

A linear estimator provides $\hatx^u$, which is an estimate of $\xu$, by applying an equalizer filter matrix $\Wbf$ to the vector of observations, $\yu$:
\begin{equation}
\begin{split}
\hatxu &= \Wbf^{H} \yu\\
&= \sum_{m=1}^{M} \w_{m}^{*} \ymu,\\
\end{split}
\label{eq:linear_det}
\end{equation}
where $\Wbf = [\w_{1} \; \w_{2} \, \cdots \, \w_{M}]^{T}$ is an $M \times K$ matrix, $\w_{m}$ is a $K \times 1$ filter vector related to antenna $m$ and $\ymu$ the observation at antenna $m$. As it can be seen the estimate $\hatxu$ is computed by the sum of $M$ partial products. If $\w_{m}$ is obtained and stored locally in the m$th$ antenna module, then the partial products can be computed with local data only, reducing the amount of data to exchange between nodes. From implementation point of view, the linear estimator relies on the accumulation of all partial results according to \eqref{eq:linear_det}, which can be done centrally (fusion node) or distributed.

For downlink, the data vector intended to the users, $\xd$, is precoded with matrix $\Pbf$ as
\begin{equation}
\yd = \Pbf\xd,\\
\label{eq:linear_prec}
\end{equation}
where $\Pbf = [\p_{1} \; \p_{2} \, \cdots \, \p_M]^{T}$ is an $M \times K$ matrix, which fulfills a power constraint $\|\Pbf\|_{F}^{2}\leq P$, such that $P$ is the maximum transmitted power. Particularly for antenna $m$ we have
\begin{equation}
\ymd = \p_{m}^T \xd.\\
\label{eq:linear_prec_i}
\end{equation}
Similarly to uplink, if $\p_{m}$ is obtained and stored locally at the m$th$ antenna module, then $\ymd$ can be computed only with local data after $\xd$ is broadcasted to all antennas.

The zero-forcing (ZF) equalizer, which is one type of linear estimator, constitutes a reference in our analysis. It is defined for uplink estimation as
\begin{equation}
\Wbf_\text{ZF}^{H} = (\Hbf^H \Hbf)^{-1}\Hbf^H,
\label{eq:W_ZF}
\end{equation}
and $\Pbf_\text{ZF}=\Wbf_\text{ZF}^{*}$ for the downlink precoding.

ZF is able to completely cancel inter-user interference (IUI) and reach the promised spectral efficiency of Massive MIMO. However, as ZF is performed in a central processor, the Gramian matrix $\Hbf^{H}\Hbf$ needs to be collected and inverted, which increases the average inter-connection data-rate. The computational load is also increased due to the matrix inversion and posterior matrix multiplication during estimation phase. Taking this into consideration, we look for methods with IUI-cancellation capabilities but with lower requirements for the system.

\subsection{Uplink \& Downlink reciprocity}
Substituting \eqref{eq:ul_model} into \eqref{eq:linear_det} leads to
\begin{equation}
\begin{split}
\hatxu
&= \Eu \xu + \zu\\
\end{split}
\label{eq:Eu}
\end{equation}
for uplink, where $\Eu = \Wbf^{H} \Hbf$ is a $K \times K$ matrix containing the equivalent uplink channel with IUI information and $\mathbf{z}^u$ is the $K \times 1$ post-equalization noise term.

On the other hand, in the downlink, substituting \eqref{eq:linear_prec} into \eqref{eq:dl_model} leads to
\begin{equation}
\begin{split}
\xdd
&= \Ed \xd + \nd,\\
\end{split}
\label{eq:Ed}
\end{equation}
where $\Ed = \Hbf^{T} \Pbf$ is a $K \times K$ matrix containing the equivalent downlink channel with IUI information. For the particular case that $\Pbf^{T} = \Wbf^{H}$, we have $\Ed = \Eu^{T}$, meaning that both equivalent channels are transposed, and therefore experiment the same IUI cancellation properties.
From this result it is clear that once an equalization matrix $\Wbf$ is obtained for uplink detection, it can also be applied for downlink precoding with no extra effort. It is interesting to note that, since $\Pbf^{T} = \Wbf^{H}$, it follows that $\p_i = \w_i^{*}$, so each antenna node can re-use same vector for detection and precoding, ideally reducing complexity and storage needs by half. Said that, in this article we focus mainly on uplink estimation without limiting the results to downlink. In reality, there is a downlink power constraint as the total transmitted power, which is addressed in \ref{section:analysis}.

\section{Centralized vs Decentralized}
\label{section:central_vs_decentral}
In this section we describe the differences between centralized and decentralized Massive MIMO processing and the justification to study the later one.

Uplink estimation based on ZF equalization has two components that should be multiplied: $\Wbf_\text{ZF}$  and $\yu$. The former includes a $K \times K$ matrix inversion, which typically is done in one place, and for that, CSI from all antennas needs to be collected. Apart from that, the observation data vector, $\yu$, is also needed for estimation. This vector is $M \times 1$, increasing considerably the amount of data to transfer and limiting the scalability of the array. Based on those considerations, we can think of two possible architectures for the Massive MIMO base-station: centralized and decentralized.

Fig. \ref{fig:BS_centralized} presents an architecture based on a central baseband processing node, where baseband samples are exchanged between Remote Processing Units (RPUs) and CPU. Each antenna is connected to a receiver and transmitter circuitry, which involves: RF front-end, ADC/DAC and OFDM processing. For simplicity, only uplink is represented in this figure. We can identify some common tasks among these processing elements across different antennas, such as: time synchronization, automatic gain control, local oscillator generation, carrier frequency and sampling rate offset estimation, phase noise compensation, among others. Therefore, a few antennas (together with corresponding receivers/transmitters) can be grouped into one RPU for efficient implementation of such common tasks. However, for simplicity, in this work we only analyze the case where each RPU manages one antenna.

Dedicated physical links would easily exceed the number of I/O connections in current standards, in addition to the increment of the cost of adding a new RPUs when needed. To overcome this, we consider that RPUs are connected to the CPU node by a shared bus as shown in Fig. \ref{fig:BS_centralized}. 

Even though, this approach can support ZF detection (and precoding) from a functionality point of view, from the implementation point of view, it requires a very high inter-connection data-rate in the bus and at the input of the CPU ($R_\mathrm{c}$ in the figure). As an example, consider a 5G NR-based system with 128 antennas and OFDM as an access technology, then the average data-rate can be calculated as
\begin{equation}
R_{\mathrm{c}} = \frac{2w M N_{\mathrm{u}}}{T_{\mathrm{OFDM}}},
\label{eq:R_central}
\end{equation}
where $N_{\mathrm{u}}$ is the number of active subcarriers, $w$ is the bit-width for the baseband samples (real/imaginary parts) after FFT, and $T_{\mathrm{OFDM}}$ is the OFDM symbol duration. For $N_{\mathrm{u}}=3300$, $w=12$ and $T_{\mathrm{OFDM}}=1/120\mathrm{kHz}$ then $R_{\mathrm{c}}=1.2 \mathrm{Tbps}$. This result clearly exceed the limit data-rate for common interfaces, such as eCPRI \cite{ecpri} and PCIe, and furthermore, it is proportional to $M$, which clearly limits the scalability of the system.

As a solution to this limitation, we propose the fully-decentralized architecture for baseband detection and precoding shown in Figure \ref{fig:BS_decentralized}. We can observe that channel estimation and estimation/precoding have been moved from the CPU to the RPUs, with detection and decoding as a remaining task in the CPU from physical layer point of view. The benefit of this move is manifold. Firstly, the inter-connection data-rate scales with $K$ instead of $M$. Secondly, the high complexity requirement in the CPU for channel estimation and data estimation/precoding is now equally distributed among RPUs, which highly simplifies the implementation and overcomes the computational bottleneck and, additionally, CSI is obtained and consumed locally in each RPU without the need for exchange, with the consequent reduction in the required inter-connection data-rate. In addition to the advantages already mentioned, which are common to other decentralized schemes, the proposed architecture presented in this work achieves an unprecedented low inter-connection data-rate by the direct connection of RPUs forming a daisy-chain, where the CPU is at one of the ends.

In the daisy-chain, depicted in Fig. \ref{fig:BS_decentralized}, nodes are connected serially to each other by a dedicated connection. All elements in the chain work simultaneously in pipeline mode, processing and transmitting/receiving to/from the respective next/previous neighbor in the chain. The data is passed through the nodes sequentially, being updated at every RPU. There is an unique connection to the root node where the last estimate is transmitted and therefore been detected by the CPU. An important remark is the average inter-connection data-rate between nodes is the same regardless of the number of elements in the chain. This topology was proposed in \cite{argos} and further studied in \cite{jesus} and \cite{muris} with specific algorithms designed for this topology.

When the decentralized architecture in Fig. \ref{fig:BS_decentralized} needs to be deployed, antennas can be collocated in the same physical place or distributed over a large area. These antennas and therefore their corresponding RPUs can behave as nodes in the chain, whilst the CPU remains as the root node. There may be multiple chains in a network. The selection of the RPUs to form a chain may depend on the users they are serving. RPUs which serve the same set of users should be in the same chain, so they can work jointly to cancel IUI. This concept fits very well with the distributed wireless communication system \cite{DWCS}, the recent cell-free Massive MIMO concept \cite{cell-free} and the promising large intelligent surface \cite{lis}.

Decentralized architectures, such as the one shown in Fig. \ref{fig:BS_decentralized}, require other type of algorithms compared to Fig. \ref{fig:BS_centralized}. In the next section we introduce our proposed algorithm, which is a method for obtaining $\w_{m}$ and $\p_{m}$ as the equalization and precoding vectors, respectively.

\section{Coordinate Descent}
\label{section:CD}

Our proposed algorithm is an iterative algorithm based on the gradient descent (GD) optimization, in which the gradient information is approximated with a set of observations in every step. From this, each antenna can obtain its own equalization/precoding vector sequentially in a coordinate descent approach. The main advantage of this method is that it does not require access to all observations at each iteration, becoming an ideal choice for large scale distributed systems.

\subsection{Preliminaries}
From \eqref{eq:Eu} we know that in the non-IUI case, $\Eu$ is a diagonal matrix, which is the case when zero-forcing (ZF) is applied. In the general case, IUI is not zero and as consequence $\Eu$ contains non-zero entries outside the main diagonal.

The objective is to find a matrix $\Wbf$, which cancels IUI to a high extent ($\Eu \approx \I$), while fulfilling the following conditions:
\begin{itemize}
	\item Uses daisy-chain as a base topology, so we exploit the advantages seen in Section \ref{section:central_vs_decentral}.
	\item No exchange of CSI between nodes. Only local CSI. 
	\item Limited amount of data to pass between antenna nodes. It should depend on $K$ instead of $M$, to enable scalability.
	\item Limit the dependency on the central processing unit in order to reduce data transfer, processing and memory requirements of that unit. One consequence of this is to avoid matrix inversion in the central unit.
\end{itemize}

\subsection{Algorithm formulation}
The algorithm setup is that one intends to solve the unconstrained Least Squares (LS) problem in the uplink
\begin{equation} \label{eq:CD_R} \hatx = \argmin_{\mathbf{x}} \|\mathbf{y}-\mathbf{H}\mathbf{x}\|^2
\end{equation}
via a GD approach. The gradient of \eqref{eq:CD_R} equals $\nabla_{\mathbf{x}}=\Hbf^{H}\Hbf\mathbf{x}-\Hbf^{H}\mathbf{y}$.
Even though $\Hbf^{H}\Hbf$ and $\Hbf^{H}\mathbf{y}$ can be formulated in a decentralized way, the selection of $\mathbf{x}$ and the product with $\Hbf^{H}\Hbf$ is preferably done in a central processing unit to limit latency and inter-connection data-rates. Following the fully-decentralized approach and the intention to off-load the equalization/precoding computation from the CPU to the RPUs, we propose a different approach.

The proposed method can be derived as an approximate version of GD that can be operated in a decentralized architecture with minimum CPU intervention. It does so by computing, at each antenna, as much as possible of $\nabla_{\mathbf{x}}$ with the information available at the antenna. Then the estimate $\hatx$ is updated by using a scaled version of the "local" gradient and the antenna passes the updated estimate on to the next antenna.

The above described procedure can, formally, be stated as
\begin{equation}
\begin{split}
\varepsilon_m &= y_{m} - \h_{m}^{T} \hatx_{m-1} \\
\hatx_{m} &= \hatx_{m-1} + \mu_m \h_{m}^{*} \varepsilon_m,
\end{split}
\label{eq:CD_sm}
\end{equation}
for antenna $m$, where $\mu_m$ is a scalar step-size. The update rule in \eqref{eq:CD_sm} corresponds to the Kaczmarz method \cite{kaczmarz}, whose step-size is according to \cite{censor}
\begin{equation}
\mu_{m} = \frac{\mu}{\|\h_{m}\|^2},
\label{eq:mu_m}
\end{equation}
where $\mu \in \mathbb{R}$ is a relaxation parameter. In case of consistent systems, this is $\mathbf{y}=\Hbf \mathbf{x}$ (if SNR is high enough or there is no noise), $\mu=1$ is optimum and the method converge to the unique solution. Otherwise, when the system is inconsistent, $\mu$ give us an extra degree of freedom, which allows to outperform the $\mu=1$ case, as we will see in Section \ref{section:analysis}.

After $M$ iterations of \eqref{eq:CD_sm} we have
\begin{equation}
\begin{split}
\hatx_{M} &= \prod_{m=1}^{M} \left( \I_K - \mu_{m} \h_{m}^{*} \h_{m}^{T} \right) \hatx_0 \\
&+ \sum_{m=1}^{M} \prod_{i=m+1}^{M} \left(\I_K - \mu_{i} \h_{i}^{*} \h_{i}^{T} \right) \mu_{m} \h_{m}^{*} y_m.
\nonumber
\end{split}
\end{equation}
If we assume $\hatx_0 = \mathbf{0}_{K\times1}$ \footnote[1]{If prior information of $\mathbf{x}$ is available, it can be used here.}, then it is possible to express $\hatx_M$ as linear combination of $\mathbf{y}$, in the same way as \eqref{eq:linear_det}, and identify
$\w_m$ (the equalization vector associated to antenna $m$) as
\begin{equation}
\w_m = \left[ \prod_{i=m+1}^{M} \left(\I_K - \mu_{i} \h_{i} \h_{i}^{H} \right) \right] \mu_{m} \h_{m}.
\label{eq:CD_W}
\end{equation}
If \eqref{eq:CD_sm} is applied in reverse antenna order ($m=M \cdots 1$), then we obtain a different estimation. The expression for $\w_{m}$ when using the alternative approach is
\begin{equation}
\w_m = \mu_{m} \A_{m-1} \h_{m},
\label{eq:CD_W2}
\end{equation}
where matrix $\A_m$ is defined as
\begin{equation}
\A_m = \prod_{i=1}^{m} \left(\I_K - \mu_{i} \h_{i} \h_{i}^{H} \right).
\label{eq:CD_A_impl}
\end{equation}

It is important to remark that both approaches lead to different $\w_{m}$ sequences, however the overall performance should be the same if CSI in all antennas shows same statistical properties (stationarity across antennas).

\subsection{Algorithm design and pseudocode}
\label{section:alg}
In this subsection we derive an equivalent and more attractive form for the calculation of the weights of the algorithm in \eqref{eq:CD_W2} in an easy and low-complexity way, suitable for hardware implementation.

The algorithm description is shown in Algorithm \ref{algo:CD}. The vector $\w_{m}$ is computed in each antenna, while the matrix $\A_{m-1}$ gets updated according to the recursive rule: $\A_{m} = \A_{m-1} - \w_{m} \h_{m}^{H}$. Then, $\w_{m}$ is stored for the detection and precoding phase, and $\A_{m}$ is passed to the next antenna node for further processing.
\IncMargin{1em}
\begin{algorithm}[ht]
	\SetKwInOut{Input}{Input}
	\SetKwInOut{Output}{Output}
	\SetKwInOut{Preprocessing}{Preprocessing}
	\SetKwInOut{Init}{Init}
	\Input{ $\Hbf = \left[ \h_{1}, \h_{2} \cdots \h_{M} \right]^{T}$}
	\Preprocessing{}
	$\A_0 = \I_K$\\
	\For{$m = 1,2,...,M$}{
		$\w_m = \mu_{m} \A_{m-1} \h_m$\\
		$\A_m = \A_{m-1} - \w_{m} \h_{m}^{H}$
	}
	\caption{Proposed algorithm}
	\label{algo:CD}
	\Output{$\Wbf = \left[ \w_{1}, \w_{2} \cdots \w_{M} \right]^{T}$}
	
\end{algorithm}\DecMargin{1em}

From Algorithm \ref{algo:CD} we can observe that after $M$ steps we achieve the following expression: $\A_M = \I_{K} - \Eu^{*}$. Then, if perfect IUI cancellation is achieved, $\Eu=\I_{K}$ and therefore $\A_{M} = \mathbf{0}$. As a consequence we can take $\|\A_{m}\|^{2}$ as a metric for residual IUI. The interpretation of Algorithm \ref{algo:CD} is as follows.  $\|\A_{m}\|$ is reduced by subtracting from $\A_{m}$ a rank-1 approximation to itself. In order to achieve that, $\A_{m}$ is projected onto $\h_{m}$ to obtain $\w_{m}$, therefore $\w_{m} \h^{H}_m$ is the best rank-1 approximation to $\A_{m}$, having $\h_{m}$ as vector base. Ideally, if the channel is rich enough, vectors $\h_{m}$ are weakly correlated and assuming $M$ is large (Massive MIMO scenario) then IUI can be canceled out to a high extent \footnote[2]{The selection of Coordinate Descent as our method's name is because we consider the vectors $\{\w_i\}$ as the outcome of the method, and these can be seen as coordinates of a cost function to minimize. Such optimization problem can be written as: $\w_{m} = \argmin_{z} f(\w_{1},\cdots,\w_{m-1},\mathbf{z},\w_{m+1},\cdots,\w_{M})$, where $f = \|\A_{m-1} - \mathbf{z} \h_{m}^{H}\|_{F}^{2}$, and $\A_{m-1} = \I_{K}-\sum_{i \neq m} \w_{i} \h_{i}^{H}$. Each antenna solves this optimization problem in a sequential fashion, obtaining one coordinate as a result, while keeping the rest fixed. This is valid for single and multiple iterations to the array, which is presented in the next subsection.}.

The role of step-size $\mu$ is to control how much IUI is removed at every iteration. High values will tend to reduce IUI faster at the beginning when the amount to remove is high, but will lead to oscillating or unstable residual IUI after some iterations because the steps are too big, so the introduced error dominates. Low values for $\mu$ will ensure convergence of the algorithm and a relatively good IUI cancellation at the expense of a slower convergence.

\subsection{Multiple-iterations along the array}
\label{sub:multiple-iter}
Recalling from Section \ref{section:alg}, Algorithm \ref{algo:CD} reduces the norm of $\A$ at each step, providing as a result $\A_{M}$, which contains the residual IUI after the algorithm is run along the array. It is possible to expand the algorithm and apply $\A_{M}$ as initial value, $\A_{0}$ for a new iteration through the array, with the intention of decreasing even more the norm of $\A$. The pseudocode of the expanded version is shown in Algorithm \ref{algo:CD_multiple}, with $n_{iter}$ iterations, and as it can be seen, an increment of $\w_{m}$ is computed at each iteration. From topology point of view, it requires an extra connection between last and first RPUs, closing the daisy-chain and becoming a ring. It is expected to improve the performance at the expense of increasing the latency.
\IncMargin{1em}
\begin{algorithm}[ht]
	\SetKwInOut{Input}{Input}
	\SetKwInOut{Output}{Output}
	\SetKwInOut{Preprocessing}{Preprocessing}
	\SetKwInOut{Init}{Init}
	\Input{ $\Hbf = \left[ \h_{1}, \h_{2} \cdots \h_{M} \right]^{T}$}
	\Preprocessing{}
	$\A_{0,1} = \I_K$\\
	$\w_{m,1} = \mathbf{0},m=1,...,M$\\
	\For{$n = 1,2,...,n_{iter}$}{
		\For{$m = 1,2,...,M$}{
			$\w_{m,n} = \w_{m,n-1} + \mu_{m} \A_{m-1,n} \h_m$\\
			$\A_{m,n} = \A_{m-1,n} - \w_{m,n} \h_{m}^{H}$\\
		}
		$\A_{0,n+1} = \A_{M,n}$
	}
	\caption{Proposed algorithm multiple iterations}
	\label{algo:CD_multiple}
	\Output{$\Wbf = \left[ \w_{1,n_{iter}}, \w_{2,n_{iter}} \cdots \w_{M,n_{iter}} \right]^{T}$}
	
	\end{algorithm}\DecMargin{1em}

\section{Analysis}
\label{section:analysis}
In this section we present an analysis of the proposed solution. The main points are:
\begin{itemize}
	\item Performance analysis of the presented solution based on SIR, SINR and BER evaluation, and comparison with other methods. 
	\item Complexity and timing analysis, including computational complexity, inter-connection throughput, memory requirement and latency.
\end{itemize}

As was commented in the Introduction, the analysis presented in this section is quite general and not dependent on any specific hardware implementation. The idea is to provide high level guidelines on algorithm-hardware trade-offs, system parameter selections, and hardware architectures. A more specific analysis can be performed when one has decided the dedicated implementation strategy.

\subsection{Performance}
\label{section:performance}
In this subsection we obtain and present different metrics to evaluate and compare the performance of the proposed algorithm. The analysis we present is divided as follows: Derivation of SIR and SINR closed form expressions, bit-error-rate (BER) analysis of the proposed algorithm based on ideal and measured channels and comparison with other methods, such as MF and ZF. The performance analysis that follows is focused on uplink, but it can be extended to downlink.

\subsubsection{SIR \& SINR}
Specifically for user $k$, \eqref{eq:Eu} is reduced to
\begin{equation}
\hat{x}^{u}_{k} = E_{k,k} x^{u}_{k} + \sum_{i=1,i \neq k}^{K} E_{k,i} x^{u}_{i} + z_{k},
\nonumber
\end{equation}
where the first term represents the desired value to estimate (scaled version), the second one is the interference from other users and the third one is due to noise. The signal-to-interference ratio (SIR) for user $k$ is defined as
\begin{equation}
\text{SIR}_{k} = \frac{\E|E_{k,k}|^{2}}{ \E \left\lbrace \sum_{i=1,i \neq k}^{K} |E_{k,i}|^2 \right\rbrace}.
\label{eq:SIR}
\end{equation}

And for the signal-to-interference-and-noise ratio (SINR) we have
\begin{equation}
\text{SINR}_{k} = \frac{\E|E_{k,k}|^{2}}{\E \left\lbrace \sum_{i=1,i \neq k}^{K} |E_{k,i}|^2 \right\rbrace + \E|z_{k}|^2 }.
\label{eq:SINR}
\end{equation}

A list of parameters and their corresponding values are presented in Table \ref{table:parameters}, which are used in the following propositions.

\begin{table}[h!]
	\begin{center}
		\caption{Parameters}
		\label{table:parameters}
		\begin{tabular}{llr}
			\cline{1-2}
			Parameter & Description \\
			\hline
			$\alpha$ &  $1-\frac{2 \mu}{K} +\frac{\mu^2}{K(K+1)}$\\
			\hline
			$\beta$ & $\frac{\mu^2}{K(K+1)}$ \\
			\hline
			$\nu$ & $1 - \frac{\mu}{K}$  \\
			\hline
			$\epsilon$ & $1 - \frac{2\mu}{K} + \frac{\mu^2}{K}$
		\end{tabular}
	\end{center}
\end{table}

From \eqref{eq:SIR} it is possible to obtain a closed-form expression of the SIR as follows:
\begin{theorem}
	\label{prop:SIR}
	With perfect $\mathrm{CSI}$ and channel model as defined in Section \ref{section:background}, $\mathrm{SIR}$ per user in uplink with $\mathrm{CD}$ algorithm for estimation is
	\begin{equation}
	\mathrm{SIR} = \frac{1 - 2\nu^{M} +  \alpha^{M} \left(1-\frac{1}{K}\right) + \epsilon^M \frac{1}{K} }{\left(1-\frac{1}{K}\right) \cdot \left( \epsilon^{M} - \alpha^{M} \right)},
	\end{equation}
	which can be simplified in case of relatively large $M$, $K$, and $\frac{M}{K}$, which is the case of Massive MIMO, as
	\begin{equation}
	\mathrm{SIR} \approx e^{\mu(2-\mu)\frac{M}{K}}.
	\label{eq:SIR_approx}
	\end{equation}
\end{theorem}
\begin{proof}
	See Appendix-\ref{proof:SIR}.
\end{proof}
\begin{figure*}\centering
	\subfloat[SINR vs $\mu$ under different SNR. M=128 and K=16.]{
		\includegraphics[width=0.48\textwidth]{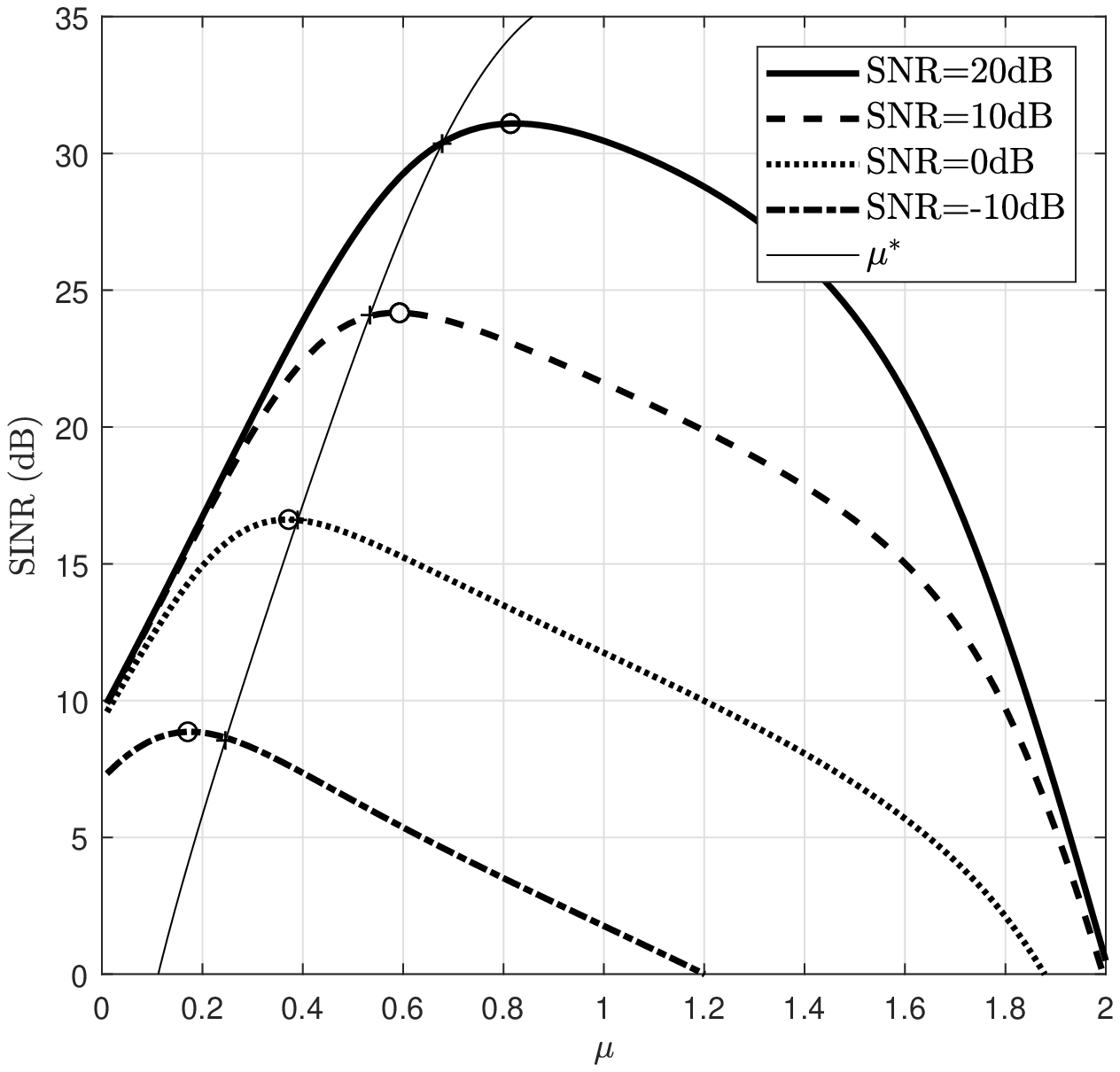}
		\label{fig:SINR_vs_mu_M128_K16}
	}
	\subfloat[SINR vs $\mu$ under different channels. M=128 and K=5. SNR=0dB.]{
		\includegraphics[width=0.48\textwidth]{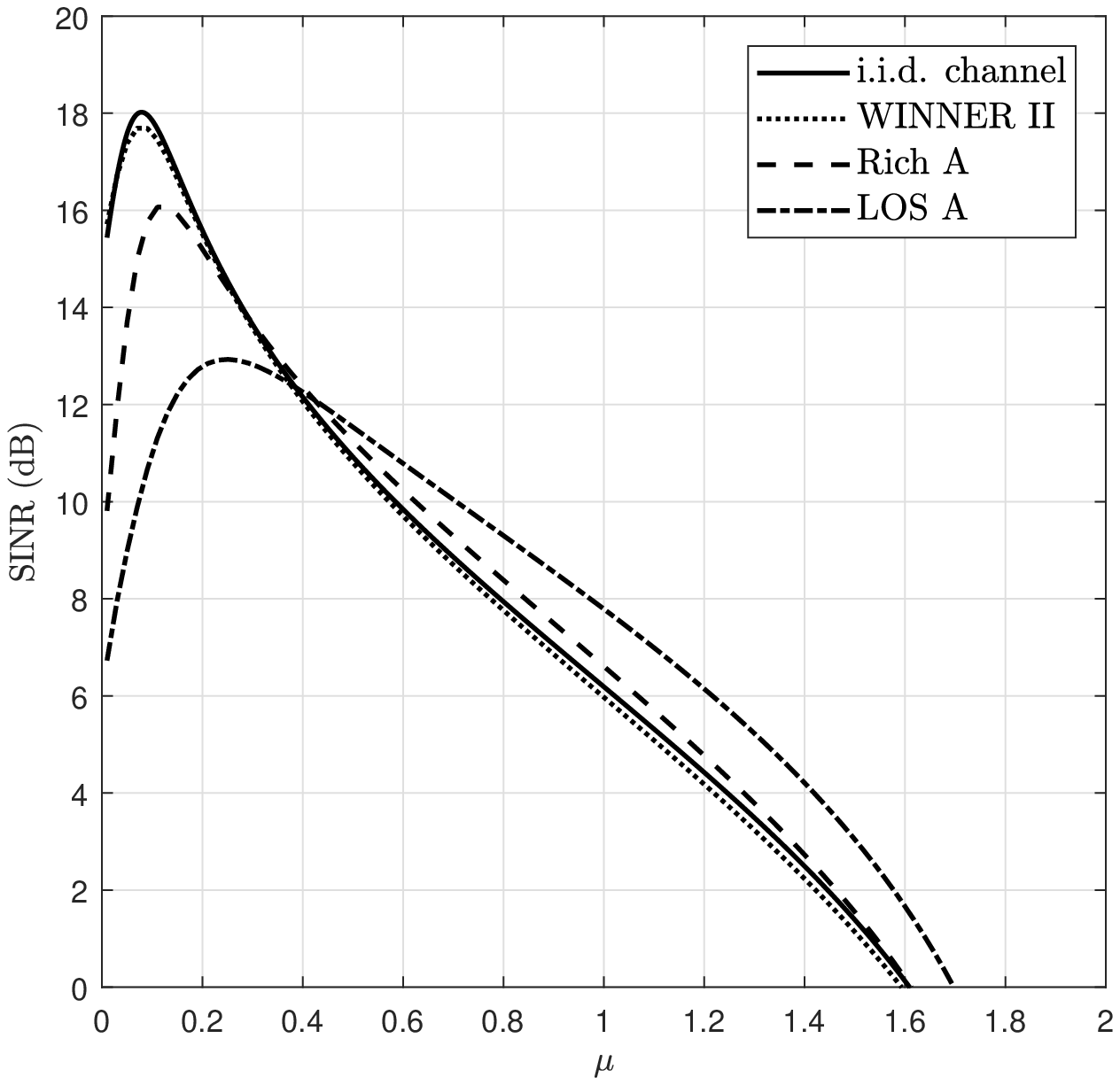}
		\label{fig:SINR_vs_mu_sim}
	}\\[-2ex]
	\caption{ }
	\label{fig:SINR_vs_mu}
	\vspace*{-4mm}
\end{figure*}

The maximum value of \eqref{eq:SIR_approx} is achieved for $\mu=1$ and the SIR value only depends on the ratio $\frac{M}{K}$ in an exponential fashion, showing how fast the IUI is canceled as $M$ grows, and therefore ZF is approached. As an example, for a target value of SIR = 40dB, $\frac{M}{K}=10$ meets the requirement, which is a typical ratio in Massive MIMO regime.

Regarding SINR, it can be derived based on previous results as
\begin{theorem}
	\label{prop:SINR}
	With perfect $\mathrm{CSI}$ and channel model as defined in Section \ref{section:background}, $\mathrm{SINR}$ per user in uplink with $\mathrm{CD}$ algorithm for estimation is given by
	\begin{equation}
	\begin{split}
	\mathrm{SINR} = \frac{1 - 2\nu^{M} +  \alpha^{M} \left(1-\frac{1}{K}\right) + \epsilon^M \frac{1}{K} }{\left(1-\frac{1}{K}\right) \left(\epsilon^{M} - \alpha^{M} \right) + \frac{N_{0}}{K-1} \left( \frac{\mu}{2-\mu}\right) (1-\epsilon^{M})},
	\end{split}
	\label{eq:SINR_CD}
	\end{equation}
	which can be simplified in case of relatively large $M$, $K$, and $\frac{M}{K}$, which is the case of Massive MIMO, as
	\begin{equation}
	\mathrm{SINR} \approx \left[ e^{-\mu(2-\mu)\frac{M}{K}} + \frac{1}{K \cdot \mathrm{SNR}} \left( \frac{\mu}{2-\mu}\right) \right]^{-1}.
	\label{eq:SINR_CD_limit}
	\end{equation}
\end{theorem}

\begin{proof}
	See Appendix-\ref{proof:SINR}.
\end{proof}

The first term in \eqref{eq:SINR_CD_limit} represents SIR containing IUI information, while the second one takes into account the post-equalized noise power. For high SNR, the first term is dominant and $\mathrm{SINR} \to e^{\mu(2-\mu)\frac{M}{K}}$, which depends on $\frac{M}{K}$ and $\mu$, but not on $\mathrm{SNR}$. On the other hand, when SNR is low, the second term is dominant and $\mathrm{SINR} \to \mathrm{SNR} \cdot K \left(\frac{2 - \mu}{\mu}\right)$ as $M$ grows, which grows linearly with $\mathrm{SNR}$ and $K$ (up to certain value). This linear dependency on $K$ is due to the post-equalization noise is equally distributed among the users. While the noise power per antenna remains constant, the portion assigned to each user decays as $K$ grows, so the SINR per user grows linearly. However, as $K$ increases the IUI does so (first term in \eqref{eq:SINR_CD_limit} grows), and both effects cancel out at some point, being IUI dominant afterwards, with the corresponding decay of SINR.

The optimal value of $\mu$, denoted as $\mu^{*}$, depends on $M$, $K$, and the specific channel. For the i.i.d. case, defined in Section \ref{section:background}, it is possible to obtain $\mu^{*}$ by numerical optimization over \eqref{eq:SINR_CD}. An approximate value, denoted as $\mu_{0}$, is presented as follows.
\begin{theorem}
	\label{prop:mu_init}
	A recommended value for $\mu_{0}$, in the vicinity of $\mu^{*}$, under CD and i.i.d. channel as defined in Section \ref{section:background}, is given by
	\begin{equation}
	\mu_{0} = \frac{1}{2} \frac{K}{M} \log (4 M \cdot \mathrm{SNR} ).
	\label{eq:mu_init}
	\end{equation}
\end{theorem}

\begin{proof}
	See Appendix-\ref{proof:mu_init}.
\end{proof}

As a side result, from the analysis performed in this section, we can extract interesting properties of the matrix $\Wbf$, such the following one:
\begin{theorem}
	\label{prop:W_power}
	The equalization matrix $\Wbf$ as result of $\mathrm{CD}$ algorithm satisfies the next properties for $\mu \in [0,2)$
	\begin{equation}
	\E \| \Wbf \|^{2}_{F} = \frac{K}{K-1} \cdot \frac{\mu}{2-\mu} \cdot \left( 1-\epsilon^{M} \right).
	\label{eq:W_power}
	\end{equation}
	\nonumber
\end{theorem}
\begin{proof}
	See Appendix-\ref{proof:W_power}.
\end{proof}

This result is relevant in downlink, where a transmission power budget is needed. Expression in \eqref{eq:W_power} is a monotonically growing function of $\mu$. It can be shown that total transmitted mean power is bounded by $4\frac{M}{K}$, reaching this value at $\mu=2$. However, as we will see in next section, optimal $\mu$ for i.i.d. Gaussian channel is within the range $(0,1]$, therefore for a large enough $K$, we have $\E \| \Wbf \|^{2}_{F} \leq 1$, which does not depend on $M$, therefore ensure the scalability of the proposed solution.\\

Expression \eqref{eq:SINR_CD} is plotted in Figure \ref{fig:SINR_vs_mu_M128_K16} showing SINR vs $\mu$ for CD under different SNR values and step-size according to \eqref{eq:mu_m}. As expected, optimal $\mu$ approaches 1 as SNR grows. Simulation results shows a good match with \eqref{eq:SINR_CD}. The curve with $\mu_{0}$ values obtained from \eqref{eq:mu_init} is also plotted for a wide range of SNR. It is observed how the $\mu_{0}$ value is reasonably close to the optimum for the SNR range depicted. Furthermore, the result is much closer to ZF than MRC values, which are $\{40.5, 30.5, 20.5, 10.5\}$dB and $\{9.0, 9.0, 8.8, 6.8\}$dB respectively for the different SNR values used in the figure.

Figure \ref{fig:SINR_vs_mu_sim} shows simulation results for the CD algorithm performance under different channels. For some of them we use a model (i.i.d and WINNER II) and others are based on real measurements (Rich A and LOS A). For this comparison we use different $\frac{M}{K}$ ratio and the step-size according to \eqref{eq:mu_m}. Rich A is non-line-of-sight (NLOS) channel, rich in scatters, while LOS A is predominantly line-of-sight (LOS) channel. WINNER II is obtained from a NLOS scenario with a uniform linear array at the BS, with M elements separated by $\lambda$/2. Users are randomly located in a 300m$\times$300m area, with the BS at the center. It is noticed how rich channels (i.i.d and WINNER II) provide better performance. SINR levels reached by ZF are \{20.9, 20.9, 19.8, 17.6\}dB and for MRC they are \{14.3, 15.2, 7.8, 4.8\}dB, in both cases for the i.i.d., WINNER II, Rich A and LOS A channels, respectively. It is also noticed that CD performance lies in between ZF and MRC for these scenarios.

Figure $\ref{fig:SINR_vs_M_over_K_SNR0}$ shows SINR versus $\frac{M}{K}$ for $M=128$ and SNR = 0dB. SINR for CD is shown comparing the effect of using  $\mu^{*}$ and $\mu_{0}$ according to \eqref{eq:mu_init}. We observe that $\frac{M}{K} \approx 10$ (equivalent to $K\approx12$) is the preferred working point, where SINR reaches the maximum value and $\mu_{0}$ gives the same result as $\mu^{*}$. We also compare the performance with ZF and MRC algorithms.

As presented in Subsection \ref{sub:multiple-iter}, the algorithm can be extended to perform multiple iterations through the array, in order to increase the performance. Figure $\ref{fig:SINR_vs_mu_sim_num_iter}$ shows SINR versus $\mu$ for a different number of iterations through the array together with ZF for comparison. From the figure we can notice that the maximum SINR increases after each iteration, approaching to ZF. It is also relevant to note that $\mu^{*}$ changes with the number of iterations. 

\begin{figure}[t]\centering
	\includegraphics[width=1\linewidth]{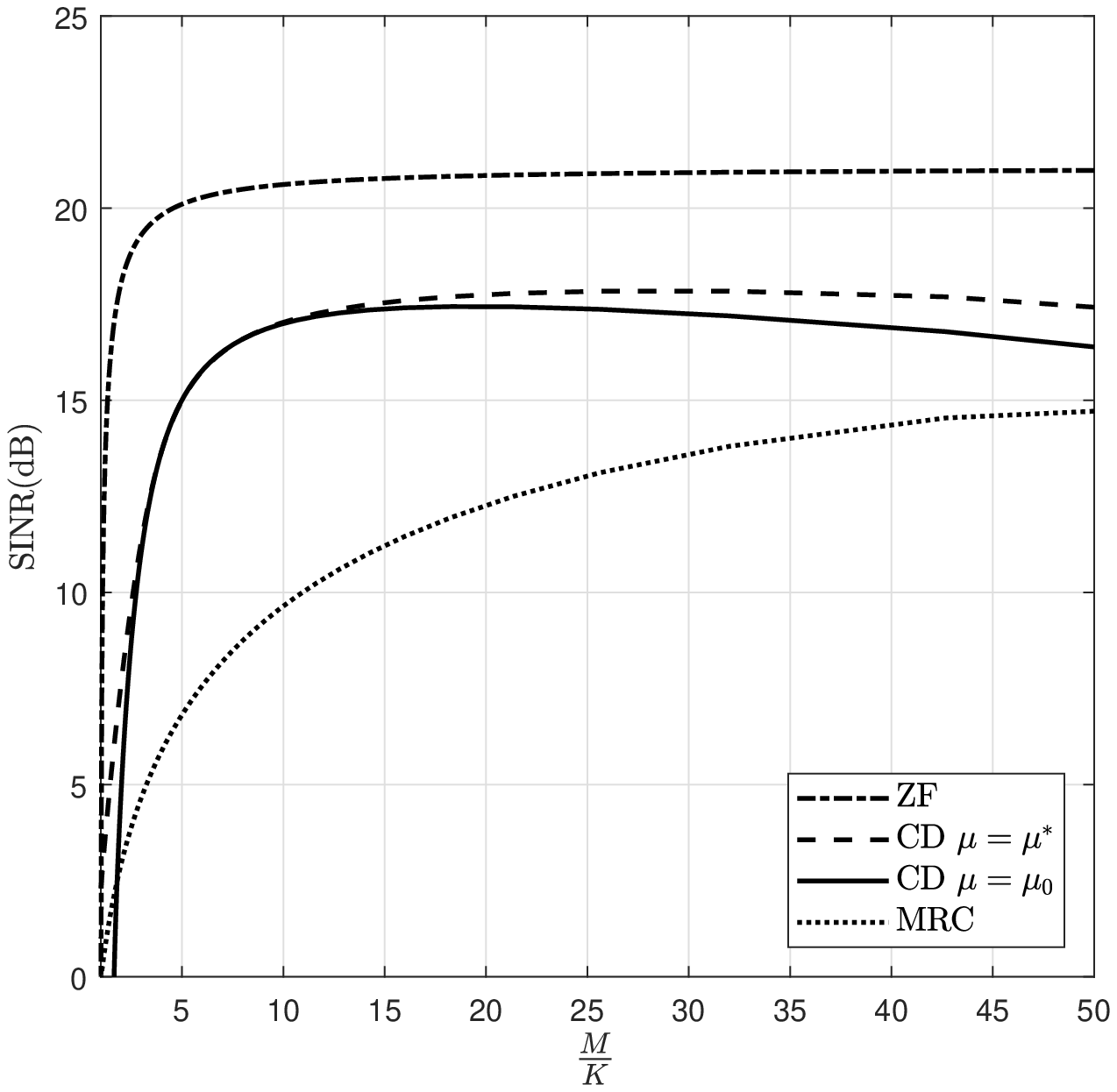}
	\vspace*{-4mm}
	\caption{SINR (dB) versus $\frac{M}{K}$ for SNR=0dB and M=128. CD SINR is plotted in the case of $\mu^{*}$ (dashed) and $\mu_{0}$ (solid) are used. i.i.d. channel. SNR = 0dB.}
	\label{fig:SINR_vs_M_over_K_SNR0}
\end{figure}

\begin{figure}\centering
	\includegraphics[width=1\linewidth]{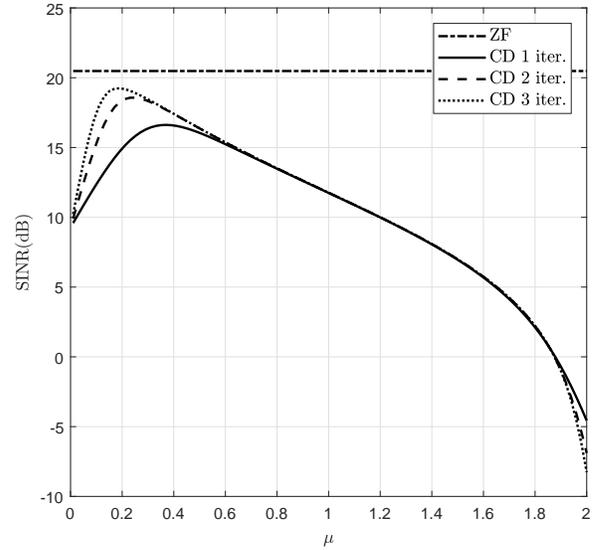}
	\vspace*{-4mm}
	\caption{SINR vs. SNR for $M$=128, $K=16$. 16QAM. i.i.d. channel. SNR=0dB. SINR after a certain number of iterations through the array. ZF added for comparison.}
	\label{fig:SINR_vs_mu_sim_num_iter}
\end{figure}

\subsubsection{BER}
BER versus SNR is shown in Figure \ref{fig:BER_vs_SNR} under i.i.d. channel for three different methods: CD, ZF and MRC. CD is shown using two different values for $\mu$: 1 and $\mu^*$. It is noticeable the great impact of the selected $\mu$ and therefore the importance of selecting an appropriate value.

The effect of non-ideal CSI in the BER is shown in Figure \ref{fig:BER_vs_SNR_non-ideal-CSI} for ZF and CD (for $\mu^{*}$). The non-ideal CSI is modeled as an ideal-CSI with a noise contribution (complex normal distributed) with a variance equal to $N_{0}$, therefore it depends inversely on SNR. No boosting in pilots is used. As it can be observed, for SNR$<$0dB the SNR gap is very small and increases as long as SNR increases too, in a similar fashion as the ideal CSI case. For SNR$>$0 the SNR gap in both cases is similar.

\begin{figure}\centering
	\includegraphics[width=1\linewidth]{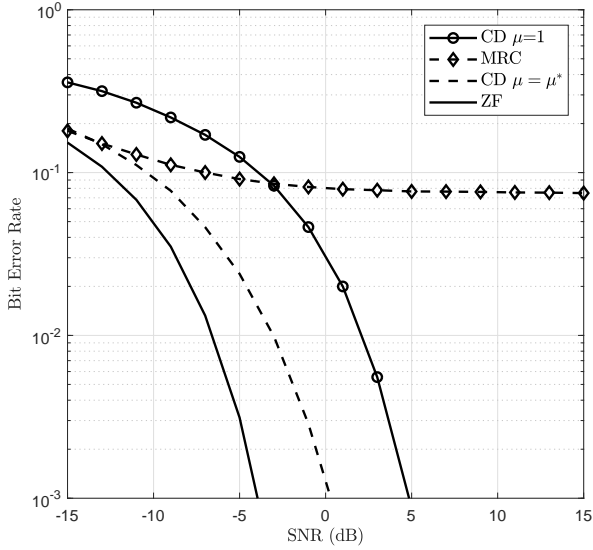}
	\vspace*{-4mm}
	\caption{BER vs. SNR for $M$=128, $K=16$. 16QAM. i.i.d. channel.}
	\label{fig:BER_vs_SNR}
\end{figure}

\begin{figure}\centering
	\includegraphics[width=1\linewidth]{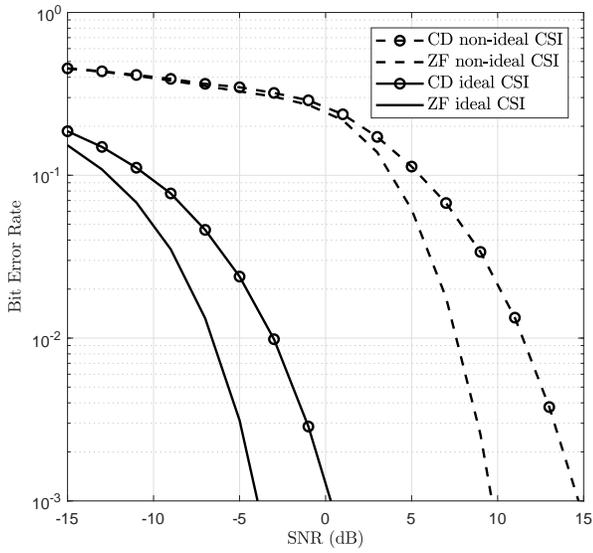}
	\vspace*{-4mm}
	\caption{BER vs. SNR for $M$=128, $K=16$. 16QAM. i.i.d. channel. Comparison between ideal and non-ideal CSI.}
	\label{fig:BER_vs_SNR_non-ideal-CSI}
\end{figure}

\subsection{Complexity \& Timing}
In this subsection we analyze the complexity of the proposed solution from three different domains: computational complexity (data processing), inter-connection throughput (data movement) and memory (data storage). Timing in the form of total system latency is also analyzed.

For this analysis we assume a frame structure based on OFDM, which contains one dedicated OFDM symbol per frame for channel estimation based on orthogonal pilots, so each one is dedicated to one of the users in a consecutive way. The other symbols convey users' data. Under the TDD assumption, some of them are used for DL and others for UL. We also assume that all RPUs perform IFFT/FFT in parallel with an output data-rate of $\frac{N_{\mathrm{u}}}{T_{\mathrm{OFDM}}}$.

\begin{figure*}[ht]
	\footnotesize
	\centering
	\psfrag{P1}{$P_1$}
	\psfrag{P2}{$P_2$}
	\psfrag{P3}{$P_3$}
	\psfrag{PN}{$P_N$}
	\psfrag{D1}{$D_1$}
	\psfrag{D2}{$D_2$}
	\psfrag{D3}{$D_3$}
	\psfrag{DN}{$D_N$}
	\psfrag{M1}{$M_1$}
	\psfrag{M2}{$M_2$}
	\psfrag{M3}{$M_3$}
	\psfrag{MN}{$M_N$}
	\psfrag{C1}{$C_1$}
	\psfrag{C2}{$C_2$}
	\psfrag{C3}{$C_3$}
	\psfrag{CN}{$C_N$}
	\psfrag{W11}{$w_{1}^{(1)}$}
	\psfrag{W12}{$w_{1}^{(2)}$}
	\psfrag{W13}{$w_{1}^{(3)}$}
	\psfrag{W1N}{$w_{1}^{(N)}$}
	\psfrag{W21}{$w_{2}^{(1)}$}
	\psfrag{W22}{$w_{2}^{(2)}$}
	\psfrag{W23}{$w_{2}^{(3)}$}
	\psfrag{W2N}{$w_{2}^{(N)}$}
	\psfrag{WM1}{$w_{M}^{(1)}$}
	\psfrag{WM2}{$w_{M}^{(2)}$}
	\psfrag{WM3}{$w_{M}^{(3)}$}
	\psfrag{WMN}{$w_{M}^{(N)}$}
	\psfrag{M1}{$M_1$}
	\psfrag{M2}{$M_2$}
	\psfrag{M3}{$M_3$}
	\psfrag{MN}{$M_N$}
	\psfrag{ant1}{$1$}
	\psfrag{ant2}{$2$}
	\psfrag{ant3}{$3$}
	\psfrag{antM}{$M$}
	\psfrag{OFDM1}{$\mathrm{OFDM} 1$}
	\psfrag{OFDM2}{$\mathrm{OFDM} 2$}
	\psfrag{OFDM3}{$\mathrm{OFDM} 3$}
	\psfrag{OFDML}{$\mathrm{OFDM} L$}
	\psfrag{A1}{$\mathbf{A}_{1}^{(n)}$}
	\psfrag{A2}{$\mathbf{A}_{2}^{(n)}$}
	\psfrag{A3}{$\mathbf{A}_{3}^{(n)}$}
	\psfrag{AM1}{$\mathbf{A}_{M-1}^{(n)}$}
	\psfrag{A11}{$\mathbf{A}_{1}^{(1)}$}
	\psfrag{A1N}{$\mathbf{A}_{1}^{(N)}$}
	\psfrag{A21}{$\mathbf{A}_{2}^{(1)}$}	
	\psfrag{A2N}{$\mathbf{A}_{2}^{(N)}$}
	\psfrag{TS}{$\cdots$}
	\psfrag{TPRB}{$T_{\mathrm{PRB}}$}
	\includegraphics[width=1.0\linewidth]{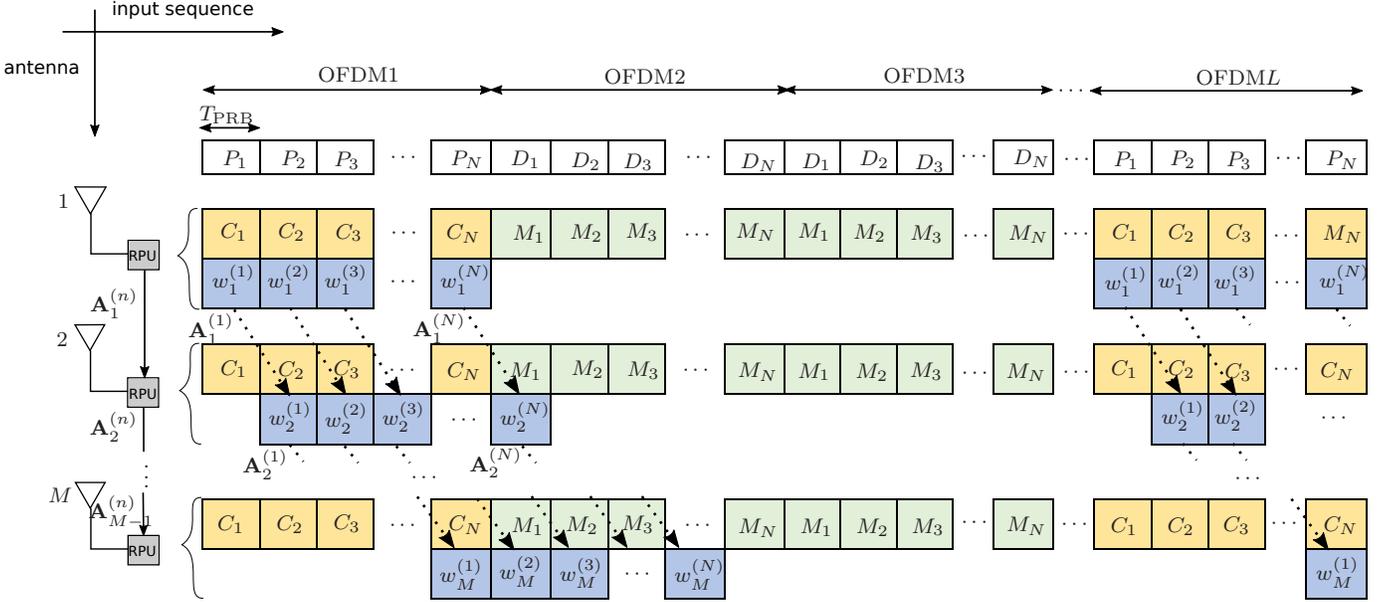}
	\caption{Time diagram representing formulation and filtering/precoding activities performed in the antenna modules. Each OFDM symbol is split into $N_\mathrm{PRB}$ blocks ($N$ in the figure) in the same order as data come out of any of the receiver FFT. Those blocks which contains pilots are shown as $P_{i}$, while those carrying data are denoted as $D_{i}$. Channel estimation is performed during $C_{i}$ blocks, while formulation is done in $\w_{i}$ blocks. Filtering/precoding data is carried out during the MIMO processing blocks, named $\mathrm{M}_{i}$. As it can be observed, all antennas perform their tasks simultaneously, while formulation is done sequentially as a matrix $\A^{(n)}$ passes through the array. In total, $N$ matrices are passed sequentially through antenna $m$, corresponding to $\A_{m}^{(n)}, n=1 \cdots N$. $\w_{i}$ vectors need to be available in the antenna modules before the corresponding data comes out of the receiver FFT so it can be properly processed. Daisy-chain topology exploits the parallelism of the operations by allowing the pipeline of the operations and the fully usage of all dedicated links simultaneously.}
	\label{fig:time_diagram}
\end{figure*}

We can exploit channel correlation based on the Physical Resource Block (PRB) concept in 3GPP. A PRB is a region in frequency-time domain where the channel response is assumed to be approximately constant across all subcarriers within that PRB. Within an OFDM symbol, the number of subcarriers in each PRB and the number of PRB per symbol, defined as $N_{\mathrm{sc,PRB}}$ and $N_{\mathrm{PRB}}$ respectively, are related as follows: $N_{\mathrm{u}} = N_{\mathrm{PRB}} N_{\mathrm{sc,PRB}}$. We define $T_{\mathrm{PRB}}$ as the time needed by $N_{\mathrm{sc,PRB}}$ consecutive subcarriers to come out the FFT.

For each PRB we have a different channel matrix and also MIMO model as in \eqref{eq:ul_model} and \eqref{eq:dl_model}. Then, it is required to have a unique set of vectors $\w_m$ and $\p_m (m=1...M)$ per antenna, as in \eqref{eq:linear_det} and \eqref{eq:linear_prec_i}, for uplink detection and downlink precoding respectively. The phase where these vectors are computed is named $\textit{formulation}$, while the phase where user's data is processed is named $\textit{filtering}$ and $\textit{precoding}$ for UL and DL respectively. To minimize data buffering, formulation needs to be completed before filtering/precoding starts. This imposes the constraint that the formulation phase needs to be finished within one OFDM symbol, or in other words, all antennas need to obtain these vectors and the matrix $\A$ needs also to pass through the array within one OFDM symbol. A diagram of the main activities involved and their timing relationship is shown in Figure \ref{fig:time_diagram}. The analysis assumes that the processing and data transmission are pipelined in each RPU so they concurrently operate.

\subsubsection{Computational complexity}

\begin{itemize}  
\item Formulation phase:
The number of complex multiplications needed to formulate one precoding/filtering vector per antenna are $C_{\mathrm{form}} \approx 2K^{2}$, which represents the matrix-vector product to obtain $\w_{m}$ and the outer product to update $\A_{m}$ according to algorithm \ref{algo:CD}. Other possible required operations such as norm, square root or division are assumed to be negligible.

\item Filtering phase:
During the filtering phase, each RPU performs the required operations  for UL detection. Vectors $\w_{m}$ are applied to all observations (data subcarriers), $y^{u}_{m}$, under the same PRB. The complexity measured in number of complex multiplications per antenna and per $N_{\mathrm{sc,PRB}}$ subcarriers is $C_{\mathrm{filt}} = KN_{\mathrm{sc,PRB}}$.

\item Precoding phase:
During the precoding phase, each RPU performs the operations required by \eqref{eq:linear_prec_i}. Similarly to the filtering case, the same vector $\p_{m}$ is applied to all data vectors $x^{d}_{m}$ under same PRB. The complexity measured in number of complex multiplications per antenna and PRB is $C_{\mathrm{prec}} = KN_{\mathrm{sc,PRB}}$.
\end{itemize}

\subsubsection{Inter-connection data-rate}
\label{section:data-rate}
\begin{itemize}  
	\item Formulation phase:
	The average inter-connection data-rate during formulation can be calculated assuming that the average time to complete a transfer of a matrix $\A$ is $T_{\mathrm{PRB}}$, which leads to an average rate of
	\begin{equation}
	R_{\mathrm{d,form}} = \frac{2w_{\A} K^{2} N_{\mathrm{PRB}}}{T_{\mathrm{OFDM}}},
	\nonumber
	\end{equation}
	where the numerator represents the amount of bits to transfer (all matrices $\A$ in a symbol) and $w_{\A}$ is the bit-width of $\A$ entries (real/imaginary parts).
	
	\item Filtering phase:
	Partial filtering results from each RPU are added up through the chain. The average inter-connection data-rate per dedicated link can be calculated as
	\begin{equation}
	R_{\mathrm{d,filt}} = \frac{2 w_{\mathrm{d}} KN_{\mathrm{u}}}{T_{\mathrm{OFDM}}},
	\nonumber
	\end{equation}
	where $w_{\mathrm{d}}$ is the bit-width of baseband samples exchanged among RPUs.
	
	\item Precoding phase:
	In the precoding phase, the data vectors $\xd$ are passed through the array for processing. Each node receives a vector which is passed to next node without any required pause (broadcasting). This leads to the same data-rate as in the filtering case.
	
\end{itemize}

\subsubsection{Latency}
The processing latency in the formulation phase for one antenna is given from next expression
\begin{equation}
\begin{split}
T_{\mathrm{proc,form}} &= \frac{C_{\mathrm{form}} T_{\mathrm{CLK}}}{N_{\mathrm{mult}}} \\
&\approx \frac{2K^{2} T_{\mathrm{CLK}}}{N_{\mathrm{mult}}},
\end{split}
\nonumber
\end{equation}
where $N_{\mathrm{mult}}$ is the number of multipliers available in each RPU that can be used in parallel, $T_{\mathrm{CLK}}$ is the clock period and we assume that one complex multiplication can be done within one  $T_{\mathrm{CLK}}$. Total latency is expressed as
\begin{equation}
\begin{split}
Lat_{form} &= M \cdot T_{\mathrm{proc, form}} + (N_{\mathrm{RPU}}-1) \cdot T_{\mathrm{trans}},
\end{split}
\nonumber
\end{equation}
where $N_{\mathrm{RPU}}$ is the number of RPUs in the system, and $T_{\mathrm{trans}}$ is the transmission latency between two consecutive RPUs. As said before, formulation needs to be finished within one $T_\mathrm{OFDM}$, therefore the formulation latency is constrained as $Lat_{form} < T_{\mathrm{OFDM}}$. This leads to an upper limit for M as
\begin{equation}
M < \frac{T_\mathrm{OFDM}+T_{\mathrm{trans}}}{T_{\mathrm{proc, form}} + \frac{T_{\mathrm{trans}}}{M_{\mathrm{RPU}}}},
\nonumber
\end{equation}
where $M_{\mathrm{RPU}}=\frac{M}{N_{\mathrm{RPU}}}$ is the number of antennas per RPU, which is considered as a design parameter. We can consider another limit, slightly lower than previous one but easier to extract conclusions as follows
\begin{equation}
M < \frac{T_\mathrm{OFDM}}{T_{\mathrm{proc, form}} + \frac{T_{\mathrm{trans}}}{M_{\mathrm{RPU}}}}.
\nonumber
\end{equation}

We analyze three scenarios:
\begin{itemize}  
	\item $T_{\mathrm{proc, form}} \rightarrow 0$: When processing time is reduced, by increasing $N_{\mathrm{mult}}$ or decreasing $T_{\mathrm{CLK}}$, then transaction time becomes dominant and a reduction in the number of links allow for higher values of $M$. Formally, the upper value for $M$ scales proportionally to $M_{\mathrm{RPU}}$ as follows
	\begin{equation}
	M < M_{\mathrm{RPU}} \cdot \frac{T_\mathrm{OFDM}}{T_{\mathrm{trans}}}.
	\nonumber
	\end{equation}
	\item $T_{\mathrm{trans}} \rightarrow 0$: By decreasing the transaction time the upper limit of $M$ converges to a certain value, which is inversely proportional to the processing time as follows
	\begin{equation}
	M < \frac{T_\mathrm{OFDM}}{T_{\mathrm{proc, form}}}.
	\nonumber
	\end{equation}
	\item $M_{\mathrm{RPU}} \gg \frac{T_{\mathrm{trans}}}{T_{\mathrm{proc, form}}}$. When $M_{\mathrm{RPU}}$ increases beyond a certain value, processing time becomes dominant and we obtain the same limit as previous point.
\end{itemize}

In case of filtering, its related processing is done in parallel as soon as data comes out of the FFT. However, partial results needs to be accumulated through the array from RPU 1 to $N_{\mathrm{RPU}}$. This latency is uniquely due to data transfer through the dedicated links, then
\begin{equation}
\begin{split}
Lat_{\mathrm{filt}} &= (N_{\mathrm{RPU}}-1) \cdot T_{\mathrm{trans}}\\
& < Lat_{\mathrm{form}}\\ & < T_{\mathrm{OFDM}}.
\end{split}
\label{eq:lat_filt}
\end{equation}

\subsubsection{Memory}
In terms of memory requirement, a centralized architecture requires to store the channel matrix $\mathbf{H}$ fully at the CPU, previous to the inversion. There is a channel matrix per PRB, so CSI storage requires $M_{\mathrm{H}} = 2 w_{\mathrm{h}} M K N_{\mathrm{PRB}}$ bits, where $w_{\h}$ represents the bit-width of $\Hbf$ entries (real/imaginary parts), and in order to store the resulting square matrix, $(\Hbf^{H}\Hbf)^{-1}$ requires $M_{\mathrm{inv}} = 2 w_{\mathrm{h}} K^{2} N_{\mathrm{PRB}}$ and therefore the total requirement is: $M_{\mathrm{central}} = M_{\mathrm{H}} + M_{\mathrm{inv}} \approx M_{\mathrm{H}}$.

In the decentralized architecture, each antenna module needs to store the corresponding $\h$, which gets replaced by $\mathbf{w}$ after formulation. Both of them requires the same amount of memory if same bit-width is assumed, which is $M_{\mathrm{w}} = 2 w_{\mathrm{h}} K N_{\mathrm{PRB}}$, and the total amount of memory in the system is: $M_{\mathrm{daisy}} = M \cdot M_{\w} \approx M_{\mathrm{central}}$. Therefore, the total amount of memory required for $\Hbf$ and $\Wbf$ is the same in both systems, however the daisy-chain allows a uniform distribution of the memory requirements across all antenna modules, reducing design complexity, time and cost. As a drawback, we point out the need for data buffering during the filtering phase due to latency in the transfer of partial results, as discussed in the previous subsection (Latency). The buffer size for the RPU closest to the CPU (worst case) can, based on \eqref{eq:lat_filt}, be obtained as
\begin{equation}
M_{\mathrm{buffer}} = \frac{ 2 w_{\mathrm{d}} K N_{\mathrm{u}} Lat_{\mathrm{filt}}}{T_{\mathrm{OFDM}}},
\nonumber
\end{equation}
which is shared by all antennas belonging to that RPU.

\subsection{Comparison}
\begin{table}
	\renewcommand{\arraystretch}{1.3} 
	\caption{Inter-connection data-rate comparison for different system parameters [$G\lowercase{b/s}$]}
	\label{tab:data-rate}
	\centering
	\begin{tabular}{l*{5}{c}}
		\hline
		Scenario    &  &   &  & \\
		$M$	   & 32     & 64     & 128    & 256 \\
		$K$	   & 4      & 8      & 12     & 12 \\
		\hline
		$R_{\mathrm{d,form}} $ & 12.67 & 50.69 & 114.05 & 114.05\\
		$R_{\mathrm{d,filt/prec}} $ & 38.02 & 76.03 & 114.05 & 114.05\\
		\hline
		$R_{\mathrm{c}} $ & 304.13 & 608.26 & 1216.51 & 2433.02\\
	\end{tabular}
\end{table}

Table \ref{tab:data-rate} shows a comparison of interconnection data-rate between daisy-chain and centralized architecture for different scenarios of $M$ and $K$. It is important to remark that $R_{\mathrm{c}}$ corresponds to the aggregated data/rate at the shared bus, while $R_{\mathrm{d}}$ is the average data/rate in each of the RPU-RPU dedicated links. For the centralized case, \eqref{eq:R_central} is used, while for the daisy-chain case, data-rates are detailed according to the different tasks (formulation, filtering and precoding) as described in Section \ref{section:data-rate}. For the numerical results we employ $T_{\mathrm{CLK}}=1\mathrm{ns}$ and $w=12$. The rest of system parameters are as follows according to worst case in 5G NR: $N_{\mathrm{u}}=3300$, $N_{\mathrm{PRB}}=275$, $N_{\mathrm{sc,PRB}}=12$ and $T_{\mathrm{OFDM}}=\frac{1}{120\mathrm{KHz}}$. We observe that for $M=128$ case, daisy-chain requires $\sim 10\%$ of the inter-connection data-rate needed by the centralized case. This number can even decrease as $\frac{M}{K}$ grows. As it is observed, daisy-chain requires much lower inter-connection data-rates than the centralized counterpart. We remark that if we take into account the total inter-connection data-rate in the decentralized case, which is $N_{\mathrm{RPU}} R_{\mathrm{d,form}}$, may easily exceed the centralized counterpart $R_{\mathrm{c}}$, however the decentralized architecture is able to distribute this data-rate equally across all links, reducing considerably the requirements for each of them.

\begin{table}
	\renewcommand{\arraystretch}{1.3} 
	\caption{Computational complexity comparison for different system parameters [$GOPS$]}
	\label{tab:complexity}
	\centering
	\begin{tabular}{l*{5}{c}}
		\hline
		Scenario    &  &   &  & \\
		$M$	   & 32     & 64     & 128    & 256 \\
		$K$	   & 4      & 8      & 12     & 12 \\
		\hline
		$C_{\mathrm{d,ant}} $ & 1.58 & 3.17 & 4.75 & 4.75\\
		\hline
		$C_{\mathrm{c}} $ & 50.69 & 202.75 & 608.26 & 1216.51\\
	\end{tabular}
\end{table}

Table \ref{tab:complexity} shows a computational complexity comparison between centralized and decentralized architectures. $C_{\mathrm{d,ant}}$ represents complex multiplications per second and per antenna in the decentralized case, while $C_{\mathrm{c}}$ is the computational complexity required by CPU in centralized system. In both cases, only filtering/precoding is taken into account because formulation depends on how often channel estimation is available. The result of the comparison is meaningful. Even tough, the total complexity in the decentralized system is approximately equal to the centralized counterpart, this is $M \cdot C_{\mathrm{d,ant}} \approx C_{\mathrm{c}}$, our decentralized solution is able to divide equally the total computational complexity among all existing RPUs, relaxing considerably the requirements compared to the CPU in centralized case. The relatively low number obtained for the daisy-chain allows the employment of cheap and general processing units in each RPU, in opposite to the centralized architecture where the total complexity requirement is on the CPU.

Numerical results for latency are shown in table \ref{tab:latency} for $N_{\mathrm{mult}}=8$, $T_{\mathrm{trans}}=100ns$ and $N_{\mathrm{RPU}}=\frac{M}{4}$. These design parameters meets the constraint $Lat < T_\mathrm{OFDM}$ up to $M=128$. For larger arrays there are different solutions: allows the latency to increase and buffer the needed input data (need for larger memory), group more antennas in each RPU (which reduces the number of links but increase the complexity of the CPU controlling each RPU), and/or employ low-latency link connections (reducing $T_{\mathrm{trans}}$ at the expense of higher cost). It is relevant to note that $T_\mathrm{OFDM}$ value in the table is the worst case $1/120KHz$.
\begin{table}
	\renewcommand{\arraystretch}{1.3} 
	\caption{Latency comparison for different system parameters}
	\label{tab:latency}
	\centering
	\begin{tabular}{l*{5}{c}}
		\hline
		Scenario    &  &   &  & \\
		$M$	   & 32     & 64     & 128    & 256 \\
		$K$	   & 4      & 8      & 12     & 12 \\
		\hline
		$Lat(\mu s)$ & 0.83 & 2.52 & 7.71 & 15.52\\
		$Lat/T_{\mathrm{OFDM}}$ & 0.10 & 0.30 & 0.92 & 1.86\\
	\end{tabular}
\end{table}

In table $\ref{tab:memory}$ a comparison between both systems from memory perspective is shown. If $w_{\h}=12$ and $N_{\mathrm{PRB}}=275$ are assumed, then for the $M=128$ case, each antenna module in the daisy-chain only needs $\sim 80 \mathrm{kbits}$ of memory and each RPU needs at maximum $354 \mathrm{kbits}$ for buffering, while in the centralized architecture, the central processor requires $\sim 11 \mathrm{Mbits}$, which is a challenging number for a cache memory. The memory requirement grows proportionally to M in the centralized system, while that does not happen in $M_{\w}$. In order to reduce the buffer size we can group more antennas in each RPU, so all of them share the same buffer memory.
\begin{table}
	\renewcommand{\arraystretch}{1.3} 
	\caption{Memory requirement comparison for different system parameters [$kbits$]}
	\label{tab:memory}
	\centering
	\begin{tabular}{l*{5}{c}}
		\hline
		Scenario    &  &   &  & \\
		$M$	   & 32     & 64     & 128    & 256 \\
		$K$	   & 4      & 8      & 12     & 12 \\
		\hline
		$M_{\w} (ant) $ & 26.4 & 52.8 & 79.2 & 79.2\\
		$M_{\mathrm{buffer}} (RPU) $ & 26.6 & 114.1 & 353.6 & 718.5\\	
		\hline
		$M_{\mathrm{H}}  $ & 844.8 & 3379.2 & 10137.6 & 20275.2\\
		$M_{\mathrm{inv}}  $ & 105.6 & 422.4 & 950.4 & 950.4\\
	\end{tabular}
\end{table}

\section{Conclusions}
\label{section:conclusions}

In this article we proposed an architecture for Massive MIMO base-station for uplink detection and downlink precoding, which is based on the fully distribution of the required baseband processing across all antenna modules in the system. The main goal is to reduce the inter-connection data-rate needed to carry out the processing tasks and enable the scalability needed in Massive MIMO. We continued our previous work in this topic \cite{jesus} \cite{muris} by a detailed introduction to the CD algorithm and its application to the Massive MIMO case. We also presented an extensive analysis of the expected performance of the system, the inter-connection data-rate, complexity, latency and memory requirements. The results show that there is a performance loss compared to ZF, but unlike MF, our proposed method does not have an error floor, from which we can not recover, while the inter-connection data-rate is distributed avoiding the aggregation of the centralized approach. At the same time, complexity and memory requirements per antenna module are easy to meet with commercial off-the-self hardware, which proves the scalability of this solution.

\appendix

In the appendix we present two propositions which are going to support the proof of propositions \ref{prop:SIR} and \ref{prop:SINR} seen in Section \ref{section:analysis}. We start with some important considerations.

Let's define the random matrix $\Q_{i}$ as
\begin{equation}
\Q_{i} \triangleq \I_K - \mu_{i} \h_{i}\h_{i}^{H},
\label{eq:Q}
\end{equation}
where $\h_{i} \sim \mathcal{CN}(0, \I)$ and independent CSI is assumed between antennas, this is $\E \{\h^{H}_{i} \h_{j}\}=\delta_{ij}, \forall i,j$. Additionally, based on \eqref{eq:Q} we can rewrite \eqref{eq:CD_A_impl} as
\begin{equation}
\A_m = \Q_{1} \Q_{2} \cdots \Q_{m},
\label{eq:A_as_prod_Q}
\end{equation}
as well as \eqref{eq:CD_W2}, which can be expressed in the following form
\begin{equation}
\w_m = \mu_{m} \Q_{1} \Q_{2} \cdots \Q_{m-1} \h_{m}.
\label{eq:CD_transformations}
\end{equation}

We list in Table \ref{table:properties} some useful properties which are used throughout this section. 

\begin{table}[h!]
	\begin{center}
		\caption{PROPERTIES}
		\label{table:properties}
		\begin{tabular}{ll}
			\hline
			$\E \left\lbrace \frac{\h\h^{H}}{\| \h \|^{2}} \right\rbrace$      & $\frac{1}{K} \I, \label{eq:hhovernormh}$   \\
			$\E \left\lbrace \frac{\h\h^{H}}{\|\h\|^{4}} \right\rbrace$  & $\frac{1}{K(K-1)} \I \label{eq:hhovernormhf}$      \\
			$\left[ \E \left\lbrace \frac{|h_{k}|^{2} \h\h^H}{\normhf} \right\rbrace \right]_{i,j} $   &  
			$\begin{cases}
				\E \left\lbrace \frac{|h_{k}|^{4}}{\normhf} \right\rbrace = \frac{2}{K(K+1)} & \text{if } k=i=j\\
				\E \left\lbrace \frac{|h_{k}|^{2} |h_{i}|^{2}}{\normhf} \right\rbrace = \frac{1}{K(K+1)} & \text{if } k \neq i=j\\
				0 & \text{if } i \neq j,\\
			\end{cases} \label{eq:h2hhovernormh4}$  \\
			$\E \{\Q_{m}\}$ & $\nu \I, \forall m$  \\
			$\E\{\A\}$ & $\nu^{M} \I \label{eq:E{A}}$ \\
		\end{tabular}
	\end{center}
\end{table}

The previous properties are based on the following proofs:
\begin{itemize}
	\item $\E \left\lbrace \frac{\h\h^{H}}{\| \h \|^{2}} \right\rbrace = a \I$, where $a$ is a complex number, due to the i.i.d. property among elements in $\h$. Applying the trace operator to both sides of previous equality, it follows that $a=\frac{1}{K}$, which proves the property.
	\item $\E \left\lbrace \frac{\h\h^{H}}{\|\h\|^{4}} \right\rbrace = a \I$, by the same principle as previous property.  Applying trace to both sides leads to: $\E \left\lbrace \frac{1}{\|\h\|^{2}} \right\rbrace = aK$. Let define the random variable $\mathbf{Y}=\|\h\|^{2}$, then $\mathbf{Y}$ follows a Chi-Square distribution with 2K-degrees of freedom, this is $\mathbf{Y} \sim \chi^{2}(2K)$, such that: $f_{\mathbf{Y}}(y) = \frac{1}{\Gamma(K)} y^{K-1} e^{-y}$. Then follows: $\E \left\lbrace \frac{1}{\mathbf{Y}} \right\rbrace = \int_{0}^{\infty} y^{-1}f_{\mathbf{Y}}(y)dy = \frac{\Gamma(K-1)}{\Gamma(K)}=\frac{1}{K-1}$, therefore: $a=\frac{1}{K(K-1)}$, and proving the property.
	\item $\E \left\lbrace \frac{|h_{k}|^{2} \h\h^H}{\normhf} \right\rbrace$ is also a diagonal matrix as previous properties. The values of the elements in the main diagonal can be obtained as follows: $1 = \E \left\lbrace \frac{\|\h\|^{4}}{\|\h\|^{4}} \right\rbrace = K \E \left\lbrace \frac{|h_{k}|^{4}}{\|\h\|^{4}} \right\rbrace + K(K-1) \E \left\lbrace \frac{|h_{k}|^{2}|h_{i}|^{2}}{\|\h\|^{4}} \right\rbrace$,
	where the next equality has been used: $\|\h\|^{4}=\sum_{k=1}^{K} |h_{k}|^{4} + \sum_{k=1}^{K} \sum_{i=1,i\neq k}^{K} |h_{k}|^{2} |h_{i}|^{2}$. Then deriving one of the expectations leads to the other one. Let's define the random variable $\mathbf{Z}=\frac{\|\h\|^{2}}{|h_{k}|^{2}}$ as: $\mathbf{Z} = 1 + \frac{\mathbf{Y}}{\mathbf{X}}$, where $\mathbf{X}=|h_{k}|^2$ and $\mathbf{Y}=\sum_{i \neq k}^{K} |h_{i}|^{2}$. $\mathbf{X}$ follows an exponential distribution, $f_{\mathbf{X}}(x)=e^{-x}$ and $\mathbf{Y} \sim \chi^{2}(2K-2)$.
	To obtain $f_{\mathbf{Z}}(z)$, first we express\\ $F_{Z}(z)=P(\mathbf{Z}\leq z)=P(\mathbf{Y}\leq \mathbf{X}(z-1))=\\ \int_{0}^{\infty} f_{\mathbf{X}}(x) \int_{0}^{x(z-1)} f_{\mathbf{Y}}(y)dydx$. The derivative with respect to $z$ is: $f_{\mathbf{Z}}(z) = \int_{0}^{\infty}f_{\mathbf{X}}(x)f_{\mathbf{Y}}(x(z-1))xdx = \frac{1}{\Gamma(K-1)}(z-1)^{K-2} \int_{0}^{\infty}x^{K-1}e^{-xz}dx \\= \frac{1}{z^{2}}(1-\frac{1}{z})^{K-2}(K-1)$ for $z\geq 1$ and $0$ otherwise, where the definition of gamma function based on improper integral has been used. Finally, $\E \left\lbrace \frac{1}{\mathbf{Z}^{2}} \right\rbrace = \int_{1}^{\infty} z^{-2}f_{\mathbf{Z}}(z)dz=\frac{2}{K(K+1)}$, proving the property.
	\item $\E \{\Q\}=\I-\mu\E\left\lbrace \frac{\h\h^{H}}{\|\h\|^{2}}\right\rbrace=\I-\mu\frac{1}{K}$, due to first property, where $\mu_m=\frac{\mu}{\|\h_{m}\|^{2}}$ has been used and the index $m$ in $\Q$ dropped for clarity.
	\item $\E \{\A\}=\E \left\lbrace \prod_{m=1}^{M} \Q_{m} \right\rbrace = \prod_{m=1}^{M} \E \Q_{m}$ due to statistical independence among antennas, then proving the property.
\end{itemize}


%
\begin{theorem}
	\label{prop:E{QDQ}}
	For a matrix $\Q$ defined as in equation \eqref{eq:Q} and $\mu$ as in \eqref{eq:mu_m}, the next result holds for any deterministic diagonal matrix $\mathbf{D}$  
	\begin{equation}
	\E \left\lbrace \Q \D \Q^{H} \right\rbrace = \alpha \mathbf{D} + \beta \Tr(\D) \I,\\
	\label{eq:E{QDQ}}
	\end{equation}
	where $\alpha$ and $\beta$ are defined in table \ref{table:parameters}.
\end{theorem}

\begin{proof}
	Let's define a deterministic diagonal matrix as $\D = diag\{d_{1},d_{2},\cdots,d_{K}\}$ and a random matrix $\Q$ defined according to \eqref{eq:Q}. Taking into account the properties in Table \ref{table:properties} we can establish the following
	\begin{equation}
	\begin{split}
	&\E \left\lbrace \Q \D \Q^{H} \right\rbrace \\&= \E \left\lbrace \D -\mu \D \hhovernormh - \mu \hhovernormh \D + \mu^{2} \frac{\h\h^H\D\h\h^H}{\normhf} \right\rbrace \\
	&= \D -2\frac{\mu}{K}\D + \mu^{2} \E \left\lbrace \frac{\h\h^H\D\h\h^H}{\normhf} \right\rbrace,
	\nonumber
	\end{split}
	\end{equation}
	where
	\begin{equation}
	\begin{split}
	\E \left\lbrace \frac{\h\h^H\D\h\h^H}{\normhf} \right\rbrace &= \E \left\lbrace \frac{\h\h^H}{\normhf} \left( \sum_{k=1}^{K} d_{k}\e_{k}\e_{k}^{T} \right) \h\h^H \right\rbrace \\
	&= \sum_{k=1}^{K} d_{k} \E \left\lbrace \frac{|\h_{k}|^{2} \h\h^H}{\normhf} \right\rbrace,
	\nonumber
	\end{split}
	\end{equation}
	which can be simplified as follows taking into account properties in table \ref{table:properties},
	\begin{equation}
	\begin{split}
	\E \left\lbrace \frac{\h\h^H\D\h\h^H}{\normhf} \right\rbrace &= \frac{1}{K(K+1)}\D + \frac{\Tr(\D)}{K(K+1)}  \I,
	\nonumber
	\end{split}
	\end{equation}
	proving the proposition.
\end{proof}

This proposition leads to a more general one.
\begin{theorem}
	\label{prop:E{ADA}}
	For a matrix $\A_{m}$ defined as in equation \eqref{eq:A_as_prod_Q} the next result holds for any deterministic diagonal matrix $\mathbf{D}$
	\begin{equation}
	\begin{split}
	\mathbb{E} \left\lbrace \A_{m} \mathbf{D} \A_{m}^{H} \right\rbrace &=
	\alpha^{m} \left[ \D - \D_{a} \right] + \epsilon^{m} \D_{a},
	\end{split}
	\label{eq:E{ADA^H}}
	\end{equation}
	where $\D_{a} = \frac{\Tr(\D)}{K}\I$, and for the particular case of $\mathbf{D}=\mathbf{I}$ it reduces to $
	\mathbb{E} \left\lbrace \A_{m} \A_{m}^{H} \right\rbrace =\epsilon^{m} \I$,
	and for $\mathbf{D} = \e_{k}^{T}\e_{k}$ the following result applies
	\begin{equation}
	\begin{split}
	\e_{k}^T \E \left\lbrace \A_{m} \e_{k}\e_{k}^{T}\A_{m}^{H} \right\rbrace \e_{k}
	&= \alpha^{m} \left(1-\frac{1}{K}\right) + \epsilon^{m} \frac{1}{K}.
	\end{split}
	\nonumber
	\end{equation}
\end{theorem}
\begin{proof}
	Let's define a sequence of diagonal matrices $\{\D_m\}_{m=0,\dots,M}$, which can be defined recursively as
	\begin{equation}
	\begin{split}
	\D_{m} = 
	\begin{cases}
	\E \left\lbrace \Q_{m} \D_{m-1} \Q_{m}^{H} \right\rbrace & \text{if } m > 0\\
	\D & \text{if } m = 0\\
	\end{cases}\\
	\end{split}
	\nonumber
	\end{equation}
	where $\Q$ is a matrix defined according to \eqref{eq:Q}.
	From proposition \ref{prop:E{QDQ}} we know that
	\begin{equation}
	\begin{split}
	\D_{m}
	&= \alpha \D_{m-1} + \beta \Tr(\D_{m-1}) \I,\\
	\end{split}
	\nonumber
	\end{equation}
	and therefore $\Tr(\D_{m}) = \epsilon \Tr(\D_{m-1})$, following that $\Tr(\D_{m}) = \epsilon^{m}\Tr(\D_{0})$, 
	which leads to
	\begin{equation}
	\begin{split}
	\D_{m} &= \alpha \D_{m-1} + \Tr(\D_{0}) \beta \epsilon^{m-1} \I \\
	&= \alpha^{m} \D_{0} + \Tr(\D_{0}) \beta  \epsilon^{m-1} \sum_{i=0}^{m-1} r^{i} \I, \\
	\end{split}
	\nonumber
	\end{equation}
	for $m>0$, where $r=\frac{\alpha}{\epsilon} < 1$, and finally taking into account that $\D_{m}=\E \left\lbrace \A_{m} \D_{0} \A_{m}^{H} \right\rbrace$ the proposition is proved.
\end{proof}

\subsection{Proof of Proposition \ref{prop:SIR}}
\label{proof:SIR}
We prove the proposition by derivation of analytical expressions for $\E|E_{k,k}|^{2}$ and $\E \left\lbrace\sum_{i=1,i \neq k}^{K} |E_{k,i}|^2 \right\rbrace$. From the properties shown in Table \ref{table:properties}, $\E|E_{k,k}|^{2}$ is expressed as
\begin{equation}
\begin{split}
&\E |E_{k,k}|^{2} = \E |\e_{k}^{T} \Eu \e_{k}|^{2} \\
&=1 - \e_{k}^{T} \E \{\A\} \e_{k} -\e_{k}^{T} \E\{\A^{H}\} \mathbf{e}_{k} + \e_{k}^{T}  \E \{\A\e_{k}\e_{k}^{T}\A^{H}\}\e_{k} \\
&= 1 - 2\nu^{M} + \alpha^{M} \left(1-\frac{1}{K}\right) + \epsilon^{M} \frac{1}{K}, 
\end{split}
\label{eq:num_analytical}
\end{equation}
and for the IUI term
\begin{equation}
\begin{split}
&\E \left\lbrace \sum_{i=1,i \neq k}^{K} |E_{k,i}|^2 \right\rbrace = \E \| \e_{k}^{T} \Eu \|^{2} - \E |E_{k,k}|^{2} \\
&= \e_{k}^{T} \E \{ \mathbf{AA}^{H} \} \e_{k} - \e_{k}^{T}  \E \{  \mathbf{A}\e_{k}\e_{k}^{T}\mathbf{A}^{H} \}\e_{k} \\
&= \left(1-\frac{1}{K}\right) \cdot \left( \epsilon^{M} - \alpha^{M} \right),
\label{eq:den_analytical}
\end{split}
\end{equation}
which proves the first part of the proposition.

In the limit when $M \to \infty$, if the ratio $\frac{M}{K}$ is kept constant, then $\epsilon^{M} \to e^{-\mu(2-\mu)\frac{M}{K}}$. Similarly, for $\alpha$ we have $\alpha^{M} \to e^{-2\mu\frac{M}{K}}$ under same conditions. Given that, we have that $\epsilon^M - \alpha^M \to e^{-2\mu\frac{M}{K}} \left[e^{\mu^2 \frac{M}{K}}-1\right]$. If we assume $\frac{M}{K}$ is large enough, such that $\mu^2 \frac{M}{K} \gg 0$, then $1$ is negligible in the second term (within brackets), and therefore $\left( 1-\frac{1}{K}\right)(\epsilon^M - \alpha^M) \to e^{-\mu(2-\mu)\frac{M}{K}}$.
In the numerator, assuming $\mu \frac{M}{K} \gg 0$, then $1 - 2\nu^{M} + \alpha^{M} (1-\frac{1}{K}) + \epsilon^{M} \frac{1}{K} \to 1$ when $M \to \infty$ and $\frac{M}{K}$ kept constant. Then, under previous assumptions regarding the ratio $\frac{M}{K}$, $\text{SIR} \to e^{\mu(2-\mu)\frac{M}{K}}$ when $M \to \infty$.
Based on this limit, we can establish the following approximation: $\text{SIR} \approx e^{\mu(2-\mu)\frac{M}{K}} = \text{SIRa}$ for large values of $M$. To give an idea of the validity of this approximation, we give some numerical values. For example, for $M=128$, $K=16$ and $\mu=1$ leads to $\text{SIR(dB)}=36.2dB$, while $\text{SIRa(dB)}=34.7dB$, resulting in a relative error of $4\%$. For $M=256$, $K=32$ and $\mu=1$, leads to an error of $2\%$, while the error goes down to $1\%$ for $M=512$ and $K=64$, approaching $0$ in the limit, and proving the second part of the proposition.

\subsection{Proof of Proposition \ref{prop:SINR}}
\label{proof:SINR}
The proof of the proposition is based on the corresponding proof to SIR expression (Appendix-\ref{proof:SIR}). The noise term, $\E|z_{k}|^2$, is the only term which has not been analyzed in the proof of Proposition \ref{proof:SIR}. This term can be computed as
\begin{equation}
\mathbb{E} |z_k|^{2} = N_{0}  \e_{k}^{T} \sum_{m=1}^{M} \mathbb{E} \left\lbrace \mathbf{w}_{m} \mathbf{w}^{H}_{m} \right\rbrace \e_{k}.
\label{eq:noise_power}
\end{equation}

Recalling that $\mathbf{w}_{m} = \mu \mathbf{A}_{m-1} \frac{\mathbf{h}_{m}}{\|\mathbf{h}_{m}\|^{2}}$ and taking into account properties in Table \ref{table:properties}, \eqref{eq:noise_power} can continue as
\begin{equation}
\begin{split}
\E |z_k|^{2}  &= \frac{\mu^2 N_{0}}{K(K-1)}  \e_{k}^{T} \sum_{m=1}^{M} \mathbb{E} \left\lbrace \mathbf{A}_{m-1} \mathbf{A}_{m-1}^{H} \right\rbrace \e_{k} \\
&=\frac{N_{0}}{K-1} \cdot \left( \frac{\mu}{2-\mu}\right) \cdot \left(1-\epsilon^{M}\right),
\end{split}
\label{eq:post_proc_noise}
\end{equation}
where Proposition $\ref{prop:E{ADA}}$ has been used, and shows that the post-processing noise power per user does not depend on $M$. This result, together with \eqref{eq:SINR} and Proposition \ref{prop:SIR} leads to the final expression shown in \eqref{eq:SINR_CD}, and proving the first part of the proposition.

In the limit, $\E |z_k|^{2} \to \frac{N_{0}}{K-1} \cdot \frac{\mu}{2-\mu}$ when $M \to \infty$. Based on this result, we can establish an approximation, similarly to the proof of Appendix-\ref{proof:SIR}, consisting of: $\E |z_k|^{2} \approx \frac{N_{0}}{K-1} \cdot \frac{\mu}{2-\mu}$ for large values of $M$. As an example of the validity of this approximation, let's consider a Massive MIMO scenario such as: $M=128$, $K=16$, $\mu=0.4$, and $N_{0}=1$. This leads to a relative error of $0.13\%$ for magnitudes in $dB$. This approximation, together with the one in Proposition \ref{prop:SIR}, provides \eqref{eq:SINR_CD_limit}. To check the validity for this approximation, for the same scenario as before, \eqref{eq:SINR_CD} and \eqref{eq:SINR_CD_limit} provide $16.60dB$ and $16.66dB$ respectively, which leads to a relative error of $0.36\%$. This completes the proof of current proposition.

\subsection{Proof of Proposition \ref{prop:mu_init}}
\label{proof:mu_init}
If $\eqref{eq:SINR_CD_limit}$ is denoted as $\SINRap$, then 
maximizing this value is equivalent to minimizing the inverse value, whose derivative is
\begin{equation}
\frac{\partial \SINRap^{-1}}{\partial \mu} = -2 (1-\mu) \frac{M}{K} e^{-\mu(2-\mu)\frac{M}{K}} + \frac{1}{K \cdot SNR} \frac{2}{(2-\mu)^2}
\nonumber
\end{equation}
and by setting to 0 leads to an expression which does not have closed form. However, which can be further simplified as: $4 M \cdot \mathrm{SNR} = e^{2\mu \frac{M}{K}}$, leading to \eqref{eq:mu_init} and proving the proposition.

\subsection{Proof of Proposition \ref{prop:W_power}}
\label{proof:W_power}
From \eqref{eq:noise_power} and \eqref{eq:post_proc_noise} we can derive the exact expression as
\begin{equation}
\begin{split}
\E \| \Wbf \|^{2}_{F} &= \Tr \E \left\lbrace \Wbf^{H} \Wbf \right\rbrace \\
&= \frac{K}{K-1} \cdot \frac{\mu}{2-\mu} \cdot \left( 1-\epsilon^{M} \right),
\end{split}
\end{equation}
proving the proposition.
\bibliographystyle{IEEEtran}
\bibliography{IEEEabrv,references}
\end{document}